\documentclass[11pt,draftcls,onecolumn]{article}
\usepackage{amsmath,amssymb,amsthm,epsfig,color,subfigure,graphicx,pgfplots}
\usepackage{enumerate,url,algpseudocode,algorithm,wasysym,epstopdf,balance}
\usepackage{bm,accents,mathrsfs}

\newcommand \cG{\mathcal{G}}
\newcommand \cH{\mathcal{H}}
\newcommand \cI{\mathcal{I}}
\newcommand \cL{\mathcal{L}}
\newcommand \cB{\mathcal{B}}
\newcommand \cN{\mathcal{N}}
\newcommand \cS{\mathcal{S}}
\newcommand \cT{\mathcal{T}}
\newcommand \cV{\mathcal{V}}
\newcommand \cX{\mathcal{X}}
\newcommand \cZ{\mathcal{Z}}
\newcommand \cP{\mathcal{P}}
\newcommand\ccalO{\mathcal{O}}

\newcommand \bzero{\mathbf{0}}
\newcommand \bone{\mathbf{1}}
\newcommand \bc{\mathbf{c}}
\newcommand \be{\mathbf{e}}
\newcommand \bh{\mathbf{h}}
\newcommand \bi{\mathbf{i}}
\newcommand \bo{\mathbf{o}}
\newcommand \bp{\mathbf{p}}
\newcommand \bq{\mathbf{q}}
\newcommand \br{\mathbf{r}}
\newcommand \bs{\mathbf{s}}
\newcommand \bv{\mathbf{v}}
\newcommand \bx{\mathbf{x}}
\newcommand \bz{\mathbf{z}}
\newcommand \bA{\mathbf{A}}
\newcommand \bB{\mathbf{B}}
\newcommand \bG{\mathbf{G}}
\newcommand \bH{\mathbf{H}}
\newcommand \bI{\mathbf{I}}
\newcommand \bP{\mathbf{P}}
\newcommand \bX{\mathbf{X}}
\newcommand \bY{\mathbf{Y}}
\newcommand \bV{\mathbf{V}}

\newcommand \bbv{\bar{\mathbf{v}}}
\newcommand \bbH{\bar{\mathbf{H}}}

\newcommand \hbv{\hat{\mathbf{v}}}
\newcommand \hbo{\hat{\mathbf{o}}}

\newcommand \bepsilon{\boldsymbol{\epsilon}}
\newcommand \blambda{\boldsymbol{\lambda}}
\newcommand \bomega{\boldsymbol{\omega}}
\newcommand \bmu{\boldsymbol{\mu}}
\newcommand \bSigma{\boldsymbol{\Sigma}}
\newcommand \bchi{\boldsymbol{\chi}}
\newcommand \bzeta{\boldsymbol{\zeta}}

\newcommand \bPhi{\boldsymbol{\Phi}}
\newcommand \bvarepsilon{\boldsymbol{\varepsilon}}

\newcommand{\minimize}{{\rm minimize}}
\newcommand{\maximize}{{\rm maximize}}

\DeclareMathOperator{\rank}{rank}
\DeclareMathOperator{\range}{range}

\DeclareMathOperator{\nullspace}{null}

\DeclareMathOperator{\diag}{dg}

\DeclareMathOperator{\trace}{Tr}

\DeclareMathOperator{\subjectto}{s.to}

\newtheorem{proposition}{Proposition}

\newtheorem{definition}{Definition}

\begin{document}
\title{PSSE Redux: Convex Relaxation, Decentralized, Robust, and Dynamic Approaches}

\author{
	Vassilis Kekatos, Gang Wang, Hao Zhu, and Georgios B. Giannakis 
	\thanks{
	V. Kekatos is with the Bradley Department of Electrical and Computer Engineering,
	Virginia Tech, Blacksburg, VA 24061, USA. 
		G. Wang is with the Department of Electrical and Computer Engineering, University of Minnesota, Minneapolis, MN 55455, USA, and also with the State Key Lab of Intelligent Control and Decision of Complex Systems, Beijing Institute of Technology, Beijing 100081, P. R. China. 	
		H. Zhu is with the Department of
		Electrical and Computer Engineering, University of Illinois, Urbana, IL 61801.
		G. B. Giannakis is with the Department of Electrical and Computer Engineering, University of Minnesota, Minneapolis, MN 55455, USA.
 E-mails: kekatos@vt.edu; \{gangwang,\,georgios\}@umn.edu; haozhu@illinois.edu.}
}

\maketitle

\allowdisplaybreaks

\section{Introduction}\label{sec:intro}
With the advent of digital computers, power system engineers in the 1960s tried computing the voltages at critical buses based on readings from current and potential transformers. Local personnel manually collected these readings and forwarded them by phone to a control center. Nevertheless, due to timing, modeling, and instrumentation inaccuracies, the power flow equations were always infeasible. In a seminal contribution~\cite{Schweppe70}, the statistical foundations were laid for a multitude of grid monitoring tasks, including topology detection, static state estimation, exact and linearized models, bad data analysis, centralized and decentralized implementations, as well as dynamic state tracking. Since then, different chapters, books, and review articles have nicely outlined the progress in the area; see for example~\cite{WollenbergBook,Mo00,AburExpositoBook}. The revolutionary monitoring capabilities enabled by synchrophasor units have been put forth in~\cite{PhTh08}.

This chapter aspires to glean some of the recent advances in power system state estimation (PSSE), though our collection is not exhaustive by any means. The Cram{\'e}r-Rao bound, a lower bound on the (co)variance of any unbiased estimator, is first derived for the PSSE setup. After reviewing the classical Gauss-Newton iterations, contemporary PSSE solvers leveraging relaxations to convex programs and successive convex approximations are explored. A disciplined paradigm for distributed and decentralized schemes is subsequently exemplified under linear(ized) and exact grid models. Novel bad data processing models and fresh perspectives linking critical measurements to cyber-attacks on the state estimator are presented. Finally, spurred by advances in online convex optimization, model-free and model-based state trackers are reviewed.

\textbf{Notation:} Lower- (upper-) case boldface letters denote column vectors (matrices), and calligraphic letters stand for sets. Vectors $\bzero$, $\bone$, and $\be_n$ denote respectively the all-zero, all-one, and the $n$-th canonical vectors of suitable dimensions. 
The conjugate of a complex-valued object (scalar, vector or matrix) $x$ is denoted by $x^\ast$; 
$\Re\{x\}$ and $\Im\{x\}$ are its real and imaginary parts, and $j:=\sqrt{-1}$. Superscripts $^\mathcal{T}$ and $^\mathcal{H}$ stand for transpose and conjugate-transpose, respectively, while $\trace(\bX)$ is the trace of matrix $\bX$. A diagonal matrix having vector $\bx$ on its main diagonal is denoted by $\diag(\bx)$; whereas, the vector of diagonal entries of $\bX$ is $\diag(\bX)$. The range space of $\bX$ is denoted by $\range(\bX)$; and its null space (kernel) by $\nullspace(\bX)$. The notation $\mathcal{N}(\bmu,\bSigma)$ represents the Gaussian distribution with mean $\bmu$ and covariance matrix $\bSigma$.
 
\section{Power Grid Modeling}\label{sec:model}
This section introduces notation and briefly reviews the power flow equations; for detailed exposition see e.g., \cite{AburExpositoBook}, \cite{SPM2013} and references therein. A power system can be represented by the graph $\cG=(\cB,\cL)$, where the node set $\cB$ comprises its $N_b$ buses, and the edge set $\cL$ its $N_l$ transmission lines. Given the focus on alternating current (AC) power systems, steady-state voltages and currents are represented by their single-phase equivalent phasors per unit.


A transmission line $(n,k)\in\cL$ running across buses $n,k\in\cB$ is modeled by its total series admittance $y_{nk}=g_{nk}+j b_{nk}$, and total shunt susceptance $j b_{nk}^s$. If $\cV_n$ is the complex voltage at bus $n$, the current $\cI_{nk}$ flowing from bus $m$ to bus $n$ over line $(m,n)$ is
\begin{equation}\label{eq:Ink}
\cI_{nk}=(y_{nk}+j b_{nk}^s/2)\cV_n-y_{nk}\cV_n.
\end{equation}
The current $\cI_{nm}$ coming from the other end of the line can be expressed symmetrically. That is not the case if the two buses are connected via a transformer with complex ratio $\rho_{nk}$ followed by a line, where
\begin{subequations}\label{eq:t}
\begin{align}
\cI_{nk}&=\frac{y_{nk}+j b_{nk}^s/2}{|\rho_{nk}|^2}\cV_n-\frac{y_{nk}}{\rho_{nk}^\ast}\cV_n\label{eq:t:Ink}\\
\cI_{nm}&=(y_{nk}+j b_{nk}^s/2)\cV_n-\frac{y_{nk}}{\rho_{nk}}\cV_n.\label{eq:t:Inm}
\end{align}
\end{subequations}

Kirchhoff's current law dictates that the current injected into bus $n$ is $\cI_n=\sum_{n\in\cB_n}\cI_{nk}$, where $\cB_n$ denotes the set of buses directly connected to bus $n$. If vector $\bi\in\mathbb{C}^{N_b}$ collects all nodal currents, and $\bv\in\mathbb{C}^{N_b}$ all nodal voltages, the two vectors are linearly related through the bus admittance matrix $\bY=\bG+j\bB$ as
\begin{equation}\label{eq:iYv}
\bi=\bY\bv.
\end{equation}
Similar to \eqref{eq:iYv}, line currents can be stacked in the $2N_l$-dimensional vector $\bi_f$, and expressed as a linear function of nodal voltages
\begin{equation}\label{eq:ifYfv}
\bi_f=\bY_f\bv
\end{equation}
for some properly defined $2N_l\times N_b$ complex matrix $\bY_f$ [cf.~\eqref{eq:Ink}--\eqref{eq:t}].

The complex power injected into bus $n$ will be denoted by $\cS_n:=P_n+jQ_n$. Since by definition $\cS_n=\cV_n\cI_n^\ast$, the vector of complex power injections $\bs=\bp+j\bq$ can be expressed as
\begin{equation}\label{eq:svYv}
\bs=\diag(\bv)\bi^\ast=\diag(\bv)\bY^\ast\bv^\ast.
\end{equation}
The power flowing from bus $n$ to bus $n$ over line $(m,n)$ is $\cS_{nk}=\cV_n\cI_{nk}^\ast$.

If voltages are expressed in polar form $\mathcal{V}_n=V_n e^{j\theta_n}$, the power flow equations in \eqref{eq:svYv} per real and imaginary entry can be written as
\begin{subequations}\label{eq:power flowpolar}
\begin{align}
P_n&=\sum_{n=1}^{N_b} V_nV_n\left[ G_{nk}\cos (\theta_n-\theta_n) +
B_{nk}\sin(\theta_n-\theta_n) \right]\label{eq:power flowpolar:p}\\
Q_n&=\sum_{n=1}^{N_b} V_nV_n\left[ G_{nk}\sin(\theta_n-\theta_n) -
B_{nk}\cos(\theta_n-\theta_n) \right].\label{eq:power flowpolar:q}
\end{align}
\end{subequations}
Since power injections are invariant if voltages are shifted by a common angle, the voltage phase is arbitrarily set to zero at a particular bus called the reference bus.

Alternatively to \eqref{eq:power flowpolar}, if voltages are expressed in rectangular coordinates $\cV_n=V_{r,m}+jV_{i,m}$, power injections are quadratically related to voltages
\begin{subequations}\label{eq:power flowrect}
\begin{align}
P_n&=V_{r,m}\sum_{n=1}^{N_b} (V_{r,n}G_{nk} - V_{i,n}B_{nk}) + V_{i,m}\sum_{n=1}^{N_b} (V_{i,n}G_{nk} + V_{r,n}B_{nk})\label{eq:power flowrect:p}\\
Q_n&=V_{i,m}\sum_{n=1}^{N_b} (V_{r,n}G_{nk} - V_{i,n}B_{nk}) - V_{r,m}\sum_{n=1}^{N_b} (V_{i,n}G_{nk} + V_{r,n}B_{nk}).\label{eq:power flowrect:q}
\end{align}
\end{subequations}

To compactly express \eqref{eq:power flowrect}, observe that $\cS_n^\ast=\cV_n^\ast\cI_n=(\bv^{\cH}\be_n)(\be_n^{\cT}\bi)=\bv^{\cH}\be_n\be_n^{\cT}\bY\bv$ from which it readily follows that
\begin{subequations}\label{eq:power flow}
\begin{align}
P_n&=\bv^{\cH}\bH_{P_n}\bv\label{eq:power flow:p}\\
Q_n&=\bv^{\cH}\bH_{Q_n}\bv\label{eq:power flow:q}
\end{align}
\end{subequations}
where the involved matrices are defined as
\begin{subequations}\label{eq:Hmat}
\begin{align}
\bH_{P_n}&:=\frac{1}{2}\left(\be_n\be_n^{\cT}\bY+ \bY^{\cH}\be_n\be_n^{\cT}\right)\label{eq:Hmat:p}\\
\bH_{Q_n}&:=\frac{1}{2j}\left(\be_n\be_n^{\cT}\bY- \bY^{\cH}\be_n\be_n^{\cT}\right).\label{eq:Hmat:q}
\end{align}
\end{subequations}
Similar expressions hold for the squared voltage magnitude at bus $n$:
\begin{equation}\label{eq:Vm}
V_n^2 =\bv^{\cH}\bH_{V_n}\bv,~~\text{where}~~\bH_{V_n}:=\be_n\be_n^{\cT}.
\end{equation} 
 
Realizing that a line current can also be provided as $\cI_{nk}=\be_{nk}^{\cT}\bi_f$, the power flow on line $(n,k)$ as seen from bus $n$ is expressed as $\cS_{nk}^\ast=\cV_n^\ast\cI_{nk}=(\bv^{\cH}\be_n)(\be_{nk}^{\cT}\bi_f)=\bv^{\cH}\be_n\be_{nk}^{\cT}\bY_f\bv$, from which it follows that
\begin{subequations}\label{eq:power flowl}
\begin{align}
P_{nk}&=\bv^{\cH}\bH_{P_{nk}}\bv\label{eq:power flow:pf}\\
Q_{nk}&=\bv^{\cH}\bH_{Q_{nk}}\bv\label{eq:power flow:qf}
\end{align}
\end{subequations}
where $\bH_{P_{nk}}$ and $\bH_{Q_{nk}}$ are defined by substituting $\be_n^{\cT}\bY$ and $ \bY^{\cH}\be_n$ by $\be_{nk}^{\cT}\bY_f$ and $\bY_f^{\cH}\be_{nk}$ in \eqref{eq:Hmat}, accordingly.

Equations \eqref{eq:power flow}, \eqref{eq:Vm} and \eqref{eq:power flowl} explain how power injections, flows, and squared voltage magnitudes are quadratic functions of voltage phasors as described by $\bv^{\cH}\bH_m\bv$ for certain complex $N_b\times N_b$ matrices $\bH_m$. Regardless if $\bY$ and/or $\bY_f$ are symmetric or Hermitian, $\bH_m$ are Hermitian by definition. This means that $\bH_m=\bH_m^{\cH}$, or equivalently, $\Re\{\bH_m\}^{\cT}=\Re\{\bH_m\}$, and $\Im\{\bH_m\}^{\cT}=-\Im\{\bH_m\}$. It can be easily verified that the quadratic functions can be expressed in terms of real-valued quantities as
\begin{equation}\label{eq:zr}
\bv^{\cH}\bH_m\bv=\bbv^{\cT}\bbH_m\bbv
\end{equation}
for the expanded real-valued voltage vector $\bbv:=[\Re\{\bv\}^{\cT}~\Im\{\bv\}^{\cT}]^{\cT}$, and the real-valued counterpart of $\bH_m$, namely
\begin{align}\label{eq:Hr}
\bbH_m:=\left[\begin{array}{cc}
\Re\{\bH_m\}& -\Im\{\bH_m\}\\ 
\Im\{\bH_m\} & \Re\{\bH_m\}
\end{array}\right].
\end{align}

\section{Problem Statement}\label{sec:problem}
It was seen in Section~\ref{sec:model} that given grid parameters collected in $\bY$ and $\bY_f$, all power system quantities can be expressed in terms of the voltage vector $\bv$, which justifies its term as the system state. Meters installed across the grid measure electric quantities, and forward their readings via remote terminal units to a control center for grid monitoring. Due to lack of synchronization, conventional meters cannot utilize the angle information of phasorial quantities. For this reason, legacy measurements involve phaseless power injections and flows along with voltage and current magnitudes at specific buses. The advent of the global positioning system (GPS) facilitated a precise timing signal across large geographical areas, thus enabling the revolutionary technology of synchrophasors or phasor measurement units (PMUs)~\cite{PhTh08}. Recovering bus voltages given network parameters and the available measurements constitutes the critical task of power system state estimation. This section formally states the problem, provides the Cram{\'e}r-Rao bound on the variance of any unbiased estimator, and reviews the Gauss-Newton iterations. Solvers based on semidefinite relaxation and successive convex approximations are subsequently explicated, and the section is wrapped up with issues germane to PMUs.

\subsection{Weighted Least-Squares Formulation}\label{subsec:problem:WLS}
Consider $M$ real-valued measurements $\{z_m\}_{m=1}^M$ related to the complex power system state $\bv$ through the model
\begin{equation}\label{eq:WLS1}
z_m=h_m(\bv)+\epsilon_m
\end{equation}
where $h_m(\bv):\mathbb{C}^{N_b}\rightarrow\mathbb{R}$ is a (non)-linear function of $\bv$, and $\epsilon_m$ captures the measurement noise and modeling inaccuracies. Collecting measurements and noise terms in vectors $\bz$ and $\bepsilon$ accordingly, the vector form of \eqref{eq:WLS1} reads
\begin{equation}\label{eq:WLS2}
\bz=\bh(\bv)+\bepsilon
\end{equation}
for the mapping $\bh:\mathbb{C}^{N_b}\rightarrow\mathbb{R}^{M}$. Model \eqref{eq:WLS2} is instantiated for different types of measurements next.

Traditionally, the system state $\bv$ is expressed in polar coordinates, namely nodal voltage magnitudes and angles. Then $\bh(\bv)$ maps the $2N_b$-dimensional state vector to SCADA measurements through the nonlinear equations \eqref{eq:power flowpolar}. Expressing the states in polar form has been employed primarily due to two reasons. First, the Jacobian matrix of $\bh(\bv)$ is amenable to approximations. Secondly, voltage magnitude measurements are directly related to states. Nevertheless, due to recent computational reformulations, most of our exposition models voltages in the rectangular form. Then, as detailed in \eqref{eq:zr}, the $m$-th SCADA measurement $z_m$ involves the quadratic function of the state $h_m(\bv)=\bv^{\cH}\bH_m\bv$ for a Hermitian matrix $\bH_m$. 

Expressing voltages in rectangular coordinates is computationally advantageous when it comes to synchrophasors too. As evidenced by \eqref{eq:iYv}--\eqref{eq:ifYfv}, PMU measurements feature \emph{linear mappings} $h_m(\bv)$. If PMU measurements are expressed in rectangular coordinates, the model in \eqref{eq:WLS1} simplifies to
\begin{equation}\label{eq:WLS2.5}
\bz=\bH\bv+\bepsilon
\end{equation}
for an $M\times N_b$ complex matrix $\bH$, and complex-valued noise $\bepsilon$. Following the notation of \eqref{eq:zr}--\eqref{eq:Hr}, the linear measurement model of \eqref{eq:WLS2.5} can be expressed in terms of real-valued quantities as
\begin{equation}\label{eq:WLS2.6}
\bar{\bz}=\overline{(\bH^\ast)}\bar{\bv}+\bar{\bepsilon}.
\end{equation}

The random noise vector $\bepsilon$ in \eqref{eq:WLS2} is usually assumed independent of $\bh(\bv)$, zero-mean and \emph{circularly symmetric}, that is $\mathbb{E}[\bepsilon\bepsilon^{\cH}]=\bSigma_\epsilon$ and $\mathbb{E}[\bepsilon\bepsilon^{\cT}]=\bzero$. The last assumption holds if for example the real and imaginary components of $\bepsilon$ are independent and have identical covariance matrices. This is true for a PMU measurement, where the actual state lies at the center of a \emph{spherically-shaped} noise cloud on the complex plane.

Moreover, the entries of $\bepsilon$ are oftentimes assumed uncorrelated yielding a diagonal covariance $\bSigma_\epsilon=\diag(\{\sigma_m^2\})$ with $\sigma_m^2$ being the variance of the $m$-th entry $\epsilon_m$. However, that may not always be the case. For example, active and reactive powers at the same grid location are derived as products between the readings of a current transformer and a potential transformer. Further, noise terms may be correlated between the real and imaginary parts of the same phasor in a PMU.

Adopting the weighted least-squares (WLS) criterion, power system state estimation can be formulated as
\begin{equation}\label{eq:WLS3}
\underset{\bv\in\mathbb{C}^{N_b}}{\minimize}~\|\bSigma_\epsilon^{-1/2}(\bz-\bh(\bv))\|_2^2
\end{equation}
where $\bSigma_\epsilon^{-1/2}$ is the matrix square root of the inverse noise covariance matrix. If the noise is independent across measurements, then \eqref{eq:WLS3} simplifies to
\begin{equation}\label{eq:WLS4}
\underset{\bv}{\minimize}~\sum_{m=1}^M \frac{(z_m-h_m(\bv))^2}{\sigma_m^2}.
\end{equation}
Either way, the PSSE task boils down to a (non)-linear least-squares (LS) fit. When the mapping $\bh(\bv)$ is linear or when the entries of $\bepsilon$ are uncorrelated, the measurement model in \eqref{eq:WLS2} can be prewhitened. For example, the linear measurement model $\bz=\bH\bv+\bepsilon$ can be equivalently transformed to
\begin{equation}\label{eq:prewhitened}
\bSigma_\epsilon^{-1/2}\bz = (\bSigma_\epsilon^{-1/2}\bH)\bv + \bSigma_\epsilon^{-1/2}\bepsilon
\end{equation}
so that the associated noise $\bSigma_\epsilon^{-1/2}\bepsilon$ is now uncorrelated. To ease the presentation, the noise covariance will be henceforth assumed $\bSigma_\epsilon = \bI_M$, yielding
\begin{equation}\label{eq:WLS5}
\hat{\bv}:=\arg\min_{\bv}~\sum_{m=1}^M (z_m-h_m(\bv))^2.
\end{equation}
For Gaussian measurement noise $\bepsilon\sim \cN(\bzero,\bI_M)$, the minimizer of \eqref{eq:WLS5} coincides with the maximum likelihood estimate (MLE) of $\bv$~\cite{kay93book}.

\subsection{Cram{\'e}r-Rao Lower Bound Analysis}\label{subsec:problem:CRLB}
According to standard results in estimation theory~\cite{kay93book}, the variance of any \emph{unbiased} estimator is lower bounded by the Cram{\'e}r-Rao lower bound (CRLB). Appreciating its importance as a performance benchmark across different estimators, the ensuing result shown in the Appendix derives the CRLB for any unbiased power system state estimator based on the so-termed \emph{Wirtinger's calculus} for complex analysis \cite{wirtinger}.

\begin{proposition}\label{prop:crlb}
	Consider estimating the unknown state vector $\mathbf{v}\in\mathbb{C}^{N_b}$ from the noisy SCADA data $\{z_m\}_{m=1}^M$ of \eqref{eq:WLS2}, where the Gaussian measurement error $\epsilon_m$ is independent across meters with mean zero and variance $\sigma_m^2$. The covariance matrix of any unbiased estimator $\hat{\mathbf{v}}$ satisfies
	\begin{equation}\label{eq:crlb}
	{\rm Cov}(\hat{\mathbf{v}})\succeq [\mathbf{F}^\dagger(\bv,\bv^\ast)]_{1:{N_b},1:{N_b}}
	\end{equation}
	where the Fisher information matrix (FIM) is given as
	\begin{equation}\label{eq:fim}
	\mathbf{F}(\bv,\bv^\ast){=}\left[\begin{array}{cc}
	\sum_{m=1}^M \tfrac{1}{\sigma_m^2} (\mathbf{H}_m\mathbf{v})(\mathbf{H}_m\mathbf{v})^\mathcal{H} &\sum_{m=1}^M \tfrac{1}{\sigma_m^{2}} (\mathbf{H}_m\mathbf{v})({\mathbf{H}}_m^\ast{\mathbf{v}^\ast})^\mathcal{H}\\
	\sum_{m=1}^M \tfrac{1}{\sigma_m^2}({\mathbf{H}}_m^\ast{\mathbf{v}^\ast})(\mathbf{H}_m\mathbf{v})^\mathcal{H}&
	\sum_{m=1}^M\tfrac{1}{\sigma_m^2}(\mathbf{H}^\ast_m{\mathbf{v}}^\ast)({\mathbf{H}}^\ast_m\mathbf{v}^\ast)^\mathcal{H}
	\end{array}
	\right].
	\end{equation}
	In addition, matrix $\mathbf{F}(\bv,\bv^\ast)$ has at least rank-one deficiency even when all possible SCADA measurements are available.
\end{proposition}

Although rank-deficient, the pseudo-inverse of $\mathbf{F}(\bv,\bv^\ast)$ qualifies as a valid lower bound on the mean-square error (MSE) of any unbiased estimator~\cite{tsp2001sm}. Rank deficiency of the FIM originates from the inherent voltage angle ambiguity: SCADA measurements remain invariant if nodal voltages are shifted globally by a unimodular phase constant. Fixing the angle of the reference bus waives this issue. It is also worth stressing that the CRLB in Prop.~\ref{prop:crlb} is oftentimes attainable and benchmarks the optimal estimator performance~\cite{tsp2001sm}. Having derived the CRLB for the PSSE task, our next subsection deals with PSSE solvers.

\subsection{Gauss-Newton Iterations}\label{subsec:problem:GN}
Consider for specificity model \eqref{eq:WLS2.5}, though the real-valued model in \eqref{eq:WLS2.6} or the model involving polar coordinates could be employed as well. When the noise covariance matrix $\bm{\Sigma}_{\epsilon}=\mathbf{I}_M$, the PSSE task in \eqref{eq:WLS3} reduces to the nonlinear LS problem
\begin{equation}\label{eq:NLS}
\underset{\bv\in\mathbb{C}^{N_b}}{\minimize}~\|\bz-\bh(\bv)\|_2^2
\end{equation}
for which the Gauss-Newton iterations are known to offer the ``workhorse'' solution \cite[Ch.~1]{Be99}, \cite[Ch.~2]{AburExpositoBook}. According to the Gauss-Newton method, the function $\bh(\bv)$ is linearized at a given point $\mathbf{v}^{i}\in\mathbb{C}^{N_b}$ using Taylor's expansion as
\begin{equation*}
\tilde{\mathbf{h}}(\mathbf{v},\mathbf{v}^{i}):= \mathbf{h}(\mathbf{v}^{i})+\mathbf{J}^i(\mathbf{v}-\mathbf{v}^{i})
\end{equation*}
where  $\mathbf{J}^i:=\nabla \mathbf{h}(\mathbf{v}^{i})$ is the $M \times N_b$ Jacobian matrix of $\mathbf{h}$ evaluated at $\mathbf{v}^i$, whose $(m,n)$-th entry is given by the Wirtinger derivative $\partial h_m/\partial \cV_n$; see e.g., \cite{wirtinger} for Wirtinger's calculus. The Gauss-Newton method subsequently approximates the nonlinear LS fit in \eqref{eq:NLS} with a linear one of $\tilde{\mathbf{h}}$, and relies on its minimizer to obtain the next iterate as
\begin{align}\label{eq:vk1}
\mathbf{v}^{i+1}&\in \arg\min_{\mathbf{v}}~\big\|\mathbf{z}-\tilde{\mathbf{h}}(\mathbf{v},\mathbf{v}^{i})\big\|^2\nonumber\\
&=\arg\min_{\mathbf{v}}~\big\|\mathbf{z}-\mathbf{h}(\mathbf{v}^{i})\big\|^2-2(\mathbf{v}-\mathbf{v}^{i})^\mathcal{H}(\mathbf{J}^{i})^\mathcal{H}(\mathbf{z}-\mathbf{h}(\mathbf{v}^{i}))\nonumber\\ &\qquad\qquad\quad
+(\mathbf{v}-\mathbf{v}^{i})^\mathcal{H}(\mathbf{J}^i)^\mathcal{H}\mathbf{J}^i
(\mathbf{v}-\mathbf{v}^{i}).
\end{align}
When matrix $(\mathbf{J}^i)^\mathcal{H}\mathbf{J}^i$ is invertible, $\mathbf{v}^{i+1}$ can be found in closed form as
\begin{equation}\label{eq:gn}
\mathbf{v}^{i+1}=\mathbf{v}^{i}+\big[(\mathbf{J}^i)^\mathcal{H}\mathbf{J}^i\big]^{-1}(\mathbf{J}^i)^\mathcal{H}(\mathbf{z}-\mathbf{h}(\mathbf{v}^{i})).
\end{equation}
The state estimate is iteratively updated using \eqref{eq:gn} until some stopping criterion is satisfied.

If, on the other hand, the WLS cost \eqref{eq:WLS3} is minimized, the Gauss-Newton iterations can be similarly obtained by treating $\bm{\Sigma}^{-1/2}_{\bm \epsilon}\mathbf{z}$ as $\mathbf{z}$ and $\bm{\Sigma}^{-1/2}_{\bm \epsilon}\mathbf{h}(\mathbf{v})$ as $\mathbf{h}(\mathbf{v})$
in \eqref{eq:NLS}, yielding
\begin{equation}\label{eq:gnw}  
\mathbf{v}^{i+1}=\mathbf{v}^{i}+\big[(\mathbf{J}^i)^\mathcal{H}\bm{\Sigma}^{-1}_{\bm \epsilon}\mathbf{J}^i\big]^{-1}(\mathbf{J}^i)^\mathcal{H}\bm{\Sigma}^{-1}_{\bm \epsilon}(\mathbf{z}-\mathbf{h}(\mathbf{v}^{i})).
\end{equation}  

It is well known that the pure Gauss-Newton iterations in \eqref{eq:gn} or \eqref{eq:gnw} may not guarantee convergence, which in fact largely depends on the starting point $\mathbf{v}^{0}$ \cite[Ch. 1.5]{Be99}. A common way to improve convergence and ensure descent of the cost in \eqref{eq:WLS3} consists of including a backtracking line search in \eqref{eq:gnw} to end up with
\begin{equation}\label{eq:gnb}
\mathbf{v}^{i+1}=\mathbf{v}^{i}+\mu^{i}\big[(\mathbf{J}^i)^\mathcal{H}\bm{\Sigma}^{-1}_{\bm \epsilon}\mathbf{J}^i\big]^{-1}(\mathbf{J}^i)^\mathcal{H}\bm{\Sigma}^{-1}_{\bm \epsilon}(\mathbf{z}-\mathbf{h}(\mathbf{v}^{i}))
\end{equation}
where the step size $\mu^{i}>0$ is found through the backtracking line search rule~\cite[Ch.~1.2]{Be99}. Due to its intimate relationship with ordinary gradient descent alternatives for nonconvex optimization however, this Gauss-Newton iterative procedure can be trapped by local solutions \cite[Ch.~1.5]{Be99}.
In a nutshell, the grand challenge remains to develop PSSE solvers capable of attaining or approximating the global optimum at manageable computational complexity. A few recent proposals in this direction are presented next.

\subsection{Semidefinite Relaxation}\label{subsec:problem:SDR}
A method to tackle the nonlinear measurement model that can convert the PSSE problem of \eqref{eq:WLS5} to a convex semidefinite program (SDP) has been introduced in~\cite{naps2012zhu}, 
\cite{jstsp2014zhu}. Consider first expressing each measurement in $\bz$ linearly in terms of the outer-product matrix $\bV:=\bv\bv^\cH$. In this way, the quadratic models in \eqref{eq:power flow}, \eqref{eq:Vm}, and \eqref{eq:power flowl} can be transformed to linear ones in terms of the matrix variable $\bV$. Thus, each noisy  measurement in \eqref{eq:WLS1} can be written as $z_m = \bv^\cH \bH_m \bv  +\epsilon_m =\trace(\bH_m \bV) + \epsilon_m$. Rewriting the PSSE task in \eqref{eq:WLS5} accordingly in terms of $\bV$ reduces to 
\begin{subequations}\label{SE_wls2}
\begin{align}
\hat{\bV}_1 := \arg\min_{\bV\in\mathbb{C}^{N_b\times N_b}}  ~&\sum_{m=1}^M
\big[z_m - \trace(\bH_m\bV) \big]^2 \label{SE_wls2f}\\
\subjectto ~&~ \bV \succeq \bzero,~~\text{and}~~\rank(\bV) = 1
\label{SE_wls2c}
\end{align}
\end{subequations}
where the positive semi-definite (PSD) and the rank-1 constraints jointly ensure that for any $\bV$ obeying \eqref{SE_wls2c}, there always exists a vector $\bv$ such that $\bV = \bv\bv^\cH$.

Although $z_m$ and $\bV$ are linearly related as in \eqref{SE_wls2}, nonconvexity is still present in two aspects: (i) the cost function in \eqref{SE_wls2f} has degree 4 in the entries of $\bV$; and (ii) the rank constraint in \eqref{SE_wls2c} is nonconvex. Aiming for an SDP reformulation of \eqref{SE_wls2}, Schur's complement lemma, see e.g., \cite[Appx. 5.5]{BoVa04}, can be leveraged to tightly bound each summand in \eqref{SE_wls2f} using an auxiliary variable $\chi_m >0$. Collecting all $\chi_m$'s in $\bchi \in \mathbb R^m$, the problem in \eqref{SE_wls2} can be expressed as
\begin{subequations}\label{SE_sdpo}
\begin{align}
\{\hat{\bV}_2,\;\hat{\bchi}_2\}:=\arg \min_{\bV,\;\bchi}~&~\bone^\cT \bchi \hfill~ \label{SE_sdpof}\\
\subjectto ~&~
\bV \succeq \bzero,~\text{and}~\rank(\bV) = 1 \label{SE_sdpoc}\\
&\left[\! \begin{array}{cc} \chi_m \!&\! z_m- \trace(\bH_{m} \bV) \\
z_m- \trace(\bH_{m} \bV)\! &\! 1 \end{array} \! \right] \!  \succeq \! \bzero,~\forall m.
\label{SE_sdpocc}
\end{align}
\end{subequations}
The equivalence among all three SE problems \eqref{eq:WLS5}, \eqref{SE_wls2}, and \eqref{SE_sdpo} has been shown in \cite{naps2012zhu}, where their optimal solutions satisfy:
\begin{align}\label{SE_eqv}
\hat{\bV}_1 &= \hat{\bV}_2=  \hat{\bv}\hat{\bv}^\cH, ~~~
\mathrm{and}~~  \hat \chi_{2,m} = \big[z_m - \trace(\bH_m \hat{\bV}_2)\big]^2,~\forall m.
\end{align}

The only source of nonconvexity in the equivalent SE problem of \eqref{SE_sdpo} comes from the rank-$1$ constraint. Motivated by the technique of semidefinite relaxation (SDR); see e.g., the seminal work of \cite{gwsdr95}, one can obtain the following convex SDP upon dropping the rank constraint
\begin{subequations}\label{SE_sdp}
\begin{align}
\{\hat{\bV},\;\hat{\bchi}\}:=\arg \min_{\bV,\;\bchi}~&\;\bone^\cT\bchi \label{SE_sdpf}\\
\subjectto ~&~\bV\succeq \bzero,  ~~\text{and}~~\eqref{SE_sdpocc}.\label{SE_sdpcc}
\end{align}
\end{subequations}

For the SDR-PSSE formulation in \eqref{SE_sdp}, a few assumptions have been made in \cite{jstsp2014zhu} to establish its global optimality in a specific setup.
{\it
\begin{enumerate}
\itemsep -2pt
\item
[$(\mathbf{as1})$] The graph $\cG=(\cB,\cL)$ has a tree topology;
\item
[$(\mathbf{as2})$] Every bus is equipped with a voltage magnitude meter; and
\item
[$(\mathbf{as3})$]  All measurements in $\bz$ are noise-free, that is $\bepsilon = \bzero$.
\end{enumerate}}

\begin{proposition} \label{prop:SE} 
Under  (as1)-(as3), solving the relaxed problem \eqref{SE_sdp} attains the global optimum of the original PSSE problem \eqref{SE_sdpo} or \eqref{eq:WLS5}; that is, $\rank(\hat{\bV})=1$.
\end{proposition}

Assumptions (as1)-(as3) may offer a close approximation of the realistic PSSE scenario, thanks to characteristics of transmission systems such as sparse connectivity, almost flat voltage profile, and high metering accuracy. Although they do not hold precisely in realistic transmission systems, near-optimality of the relaxed problem \eqref{SE_sdp} has been numerically supported by extensive tests~\cite{jstsp2014zhu}. A more crucial issue is to recover a feasible SE solution from the relaxed problem \eqref{SE_sdp}, as $\hat{\bV}$ is very likely to have rank greater than 1. This is possible either by finding the best rank-1 approximation to $\hat{\bV}$ via eigenvalue decomposition, or via randomization~\cite{LuMa10}.

SDR endows SE with a convex SDP formulation for which efficient schemes are available to obtain the global optimum using for example the interior-point solver. The computational complexity for eigen-decomposition is in the order of matrix multiplication, and thus negligible compared to solving the SDP; see \cite{LuMa10} and references therein. However, the polynomial complexity of solving the SDP could be a burden for real-time power system monitoring, which motivates well the distributed implementation of Section~\ref{subsec:distsolv:sdr}.

\subsection{Penalized Semidefinite Relaxation}\label{subsec:problem:PSDR}
Building on an alternative formulation of the \emph{power flow problem}, a penalized version of the aforementioned SDP-based state estimator has been devised in~\cite{MLB15,MALB16,psse2016zhang}. Commencing with the power flow task, it can be interpreted as a particular instance of PSSE, where:
\begin{itemize}
\item Measurements (henceforth termed specifications) are noiseless;
\item Excluding the reference bus, buses are partitioned into the subset $\cB_{\text{PV}}$ for which active injections and voltage magnitudes are specified, and the subset $\cB_{\text{PQ}}$, for which active and reactive injections are specified.
\end{itemize}
The power flow task can be posed as the feasibility problem; that is,
\begin{align}\label{eq:power flow0}
\textrm{find}~&~\bv\in\mathbb{C}^{N_b}\\
\subjectto ~&~P_n=\bv^{\cH}\bH_{P_n}\bv,\quad \forall n\in \cB_{\text{PV}}\cup\cB_{\text{PQ}}\nonumber\\
&~Q_n=\bv^{\cH}\bH_{Q_n}\bv,\quad \forall n\in \cB_{\text{PQ}}\nonumber\\
&~V_n^2=\bv^{\cH}\bH_{V_n}\bv,\quad \forall n\in \cB_{\text{PV}},~\text{and}~V_\text{ref}^2=V_0.\nonumber
\end{align}
Using the SDP reformulation presented earlier, the power flow task can be equivalently expressed as
\begin{align}\label{eq:power flow1}
\textrm{find}~&~\bV\in\mathbb{C}^{N_b\times N_b}\\
\subjectto ~&~z_m=\trace(\bH_{m} \bV),\quad m=1,\ldots,2N_b-1\nonumber\\
~&~\bV \succeq \bzero,~\text{and}~\rank(\bV) = 1\nonumber
\end{align}
where the specifications (constraints) of \eqref{eq:power flow0} have been generically captured by the pairs $\{(z_m;\bH_m)\}_{m=1}^{2N_b-1}$. 

Although the optimization in \eqref{eq:power flow1} is non-convex, a convex relaxation can be obtained by dropping the rank constraint. To promote rank-one solutions, the feasibility problem is further reduced to \cite{MLB15}
\begin{align}\label{eq:power flow2}
\hat{\bV}:=\arg\min_{\bV\succeq \bzero}~&\;\trace(\bH_0\bV)\\
\subjectto ~&~z_m=\trace(\bH_{m} \bV),~m=1,\ldots,2N_b-1.
\end{align}
If the Hermitian matrix $\bH_0$ is selected such that $\bH_0\succeq \bzero$, $\rank(\bH_0)=N_b-1$, and $\bH_0\bone=\bzero$; then any state $\bv$ close to the flat voltage profile $\bone+j\bzero$ can be recovered from the rank-one minimizer $\hat{\bV}$ of \eqref{eq:power flow2}; see~\cite[Thm.~2]{MLB15}. The stated conditions exclude the case $\bH_0=\bI_{N_b}$ that would have led to the nuclear-norm heuristic commonly used in low-rank matrix completion.

Spurred by this observation, the PSSE task can be posed as~\cite{MALB16}
\begin{equation}\label{eq:pPSSE1}
\hat{\bV}_{\mu}:=\arg\min_{\bV\succeq \bzero}~\trace(\bH_0\bV)+\mu\sum_{m=1}^M f_m\big(z_m-\trace(\bH_{m} \bV)\big)
\end{equation}
for some regularization parameter $\mu\geq 0$, where $M$ now can be larger than $2N_b-1$. The second term in the cost of \eqref{eq:pPSSE1} is a data-fitting term ensuring that the recovered state is consistent with the collected measurements based on selected criteria. Two cases for $f_m$ of special interest are the LS fit $f_m^{\text{LS}}(\epsilon):=\epsilon^2$, and the least-absolute value (LAV) one $f_m^{\text{LAV}}(\epsilon):=|\epsilon|$, $\forall m=1,\ldots,M$. On the other hand, the first term in \eqref{eq:pPSSE1} can be understood as a \emph{regularizer} to promote rank-one solutions $\hat{\bV}_{\mu}$; see more details in \cite{MALB16}. 

In the noiseless setup, where all measurements comply with the model $z_m=\bv^{\cH}_0\bH_m\bv_0$, the minimization in \eqref{eq:pPSSE1} has been shown to possess a rank-one minimizer $\hat{\bV}_{\mu}=\hat{\bv}_{\mu}\hat{\bv}_{\mu}^{\cH}$ for all $\mu\geq 0$ under both $f_m^{\text{LS}}$ and $f_m^{\text{LAV}}$; see details in \cite{MALB16}. Interestingly, the solution $\hat{\bv}_{\mu}$ obtained under the LS fit does not coincide with $\bv_0$ for any $\mu\geq 0$, whereas the LAV solution $\hat{\bv}_{\mu}$ provides the actual state $\bv_0$ for a sufficiently large $\mu$. Error bounds between $\hat{\bV}_{\mu}$ and $\bv_0\bv_0^{\cH}$ under the regularized LAV solution for noisy measurements are established in~\cite{psse2016zhang}.

\subsection{Feasible Point Pursuit}\label{subsec:problem:FPP}
The feasible point pursuit (FPP) method studied in~\cite{spl2015sidiropoulos} offers another computationally manageable solver for approximating the globally optimal PSSE. 
As a special case of the convex-concave procedure \cite{ccp2016}, 
FPP is an iterative algorithm for handling general nonconvex quadratically constrained quadratic programs (QCQPs)~\cite{tsg2016zamzam}. It approximates the feasible solutions of a nonconvex QCQP by means of a sequence of convexified QCQPs obtained with successive convex inner-restrictions of the original nonconvex feasibility set~\cite{tsg2016zamzam}. 

The first step in applying FPP to PSSE is a reformulation of \eqref{eq:WLS3} into a standard QCQP~\cite{tps2017wzgs,globalsip2016wzgs}
\begin{subequations}\label{eq:qcqpstd}
\begin{align}
\underset{\mathbf{v}\in\mathbb{C}^{N_b},\;\bchi\in\mathbb{R}^M}{\minimize}~&~\bchi^\mathcal{T}\bm{\Sigma}_{\bm \epsilon}^{-1}\bchi\\
\subjectto~&~\mathbf{v}^\mathcal{H} \mathbf{H}_{m} \mathbf{v} \leq  z_m   + \chi_m,\quad\quad\quad\;\;  1\le m \le M\label{eq:qcqpstd2}\\
~&~\mathbf{v}^\mathcal{H} \left(-\mathbf{H}_{m}\right) \mathbf{v} \leq  -z_m   + \chi_m,\quad 1\le m\le M \label{eq:qcqpstd3}
\end{align}
\end{subequations}
where vector $\bchi\in\mathbb{R}^M$ collects the auxiliary variables $\{\chi_m\ge 0\}_{m=1}^M$. For power flow and power injection measurements, the corresponding Hermitian measurement matrices $\{\mathbf{H}_m\}$ are indefinite in general; thus, both constraints \eqref{eq:qcqpstd2} and \eqref{eq:qcqpstd3} are nonconvex. On the contrary, squared voltage magnitude measurements relate to positive semidefinite matrices $\{\mathbf{H}_m\}$, so that only the related constraint \eqref{eq:qcqpstd3} is nonconvex. Either way, problem \eqref{eq:qcqpstd} is NP-hard, and hence computationally intractable~\cite{nphard}. 

Using eigen-decomposition, every measurement matrix $\mathbf{H}_m$ can be expressed as the sum of a positive and a negative semidefinite matrix as $\mathbf{H}_m= \mathbf{H}_m^{+}+\mathbf{H}_m^{-}$, so that the constraints in \eqref{eq:qcqpstd} are rewritten as
\begin{subequations}\label{eq:con}
\begin{align}
~&~~	\mathbf{v}^\mathcal{H} \mathbf{H}_{m}^{+} \mathbf{v} + \mathbf{v}^\mathcal{H} \mathbf{H}_{m}^{-} \mathbf{v}\leq  z_m   + \chi_m\\
~&~~\mathbf{v}^\mathcal{H} \mathbf{H}_{m}^{+} \mathbf{v} +\mathbf{v}^\mathcal{H} \mathbf{H}_{m}^{-} \mathbf{v} \geq  z_m   - \chi_m
\end{align}
\end{subequations}
for $m=1,\ldots,M$. Observe now that since $\mathbf{v}^\mathcal{H} \mathbf{H}_{m}^{-} \mathbf{v}$ is a concave function of $\bv$, it is upper bounded by its first-order (linear) approximation at any point $\bv^i$; that is,
\[\mathbf{v}^\mathcal{H} \mathbf{H}_{m}^{-} \mathbf{v}\leq 2\Re\{(\mathbf{v}^i)^\mathcal{H} \mathbf{H}_{m}^{-} \mathbf{v}\}  - (\mathbf{v}^i)^\mathcal{H} \mathbf{H}_{m}^{-} \mathbf{v}^i.\]
The concave function $\mathbf{v}^\mathcal{H} (-\mathbf{H}_{m}^{+})\mathbf{v}$ can be upper bounded similarly. 

Based on this observation, the FPP technique replaces the concave functions in the constraints of \eqref{eq:con} with their linear upper bounds. The point of linearization at every iteration is chosen to be the previous state estimate. Specifically, initializing with $\mathbf{v}^{0}$, the FPP produces the iterates
\begin{align}\label{eq:fpp}
\{\mathbf{v}^{i+1},\;\bchi^{i+1}\}&:=\arg\min_{\mathbf{v},\;\bchi\geq \bzero}~\bchi^\mathcal{T}\bm{\Sigma}_{\bm \epsilon}^{-1}\bchi\\
\subjectto~&\;\mathbf{v}^\mathcal{H} \mathbf{H}_m^{+} \mathbf{v} + 2 \Re\{(\mathbf{v}^{i})^\mathcal{H} \mathbf{H}_m^{-} \mathbf{v}\}  \leq z_m + (\mathbf{v}^{i})^\mathcal{H}\mathbf{H}_m^{-} \mathbf{v}^{i} + \chi_m,~\forall m\nonumber\\
&\;\mathbf{v}^\mathcal{H} \mathbf{H}_m^{-} \mathbf{v} + 2 \Re\{(\mathbf{v}^{i})^\mathcal{H} \mathbf{H}_m^{+} \mathbf{v}\}  \geq z_m + (\mathbf{v}^{i})^\mathcal{H} \mathbf{H}_m^{+} \mathbf{v}^{i} -\chi_m,~\forall m.\nonumber
\end{align}
At every iteration, the FPP technique solves the now convex QCQP in \eqref{eq:fpp}. The procedure has been shown to globally converge to a stationary point of the WLS formulation \eqref{eq:WLS3} of the PSSE task~\cite{tps2017wzgs}. Extensions of the developed FPP solver to cope with linear(ized) measurements and bad data are straightforward; see \cite{tps2017wzgs} for details.

\begin{figure}[t]
\begin{subfigure}
\centering
\includegraphics[width=60.5mm,scale=2]{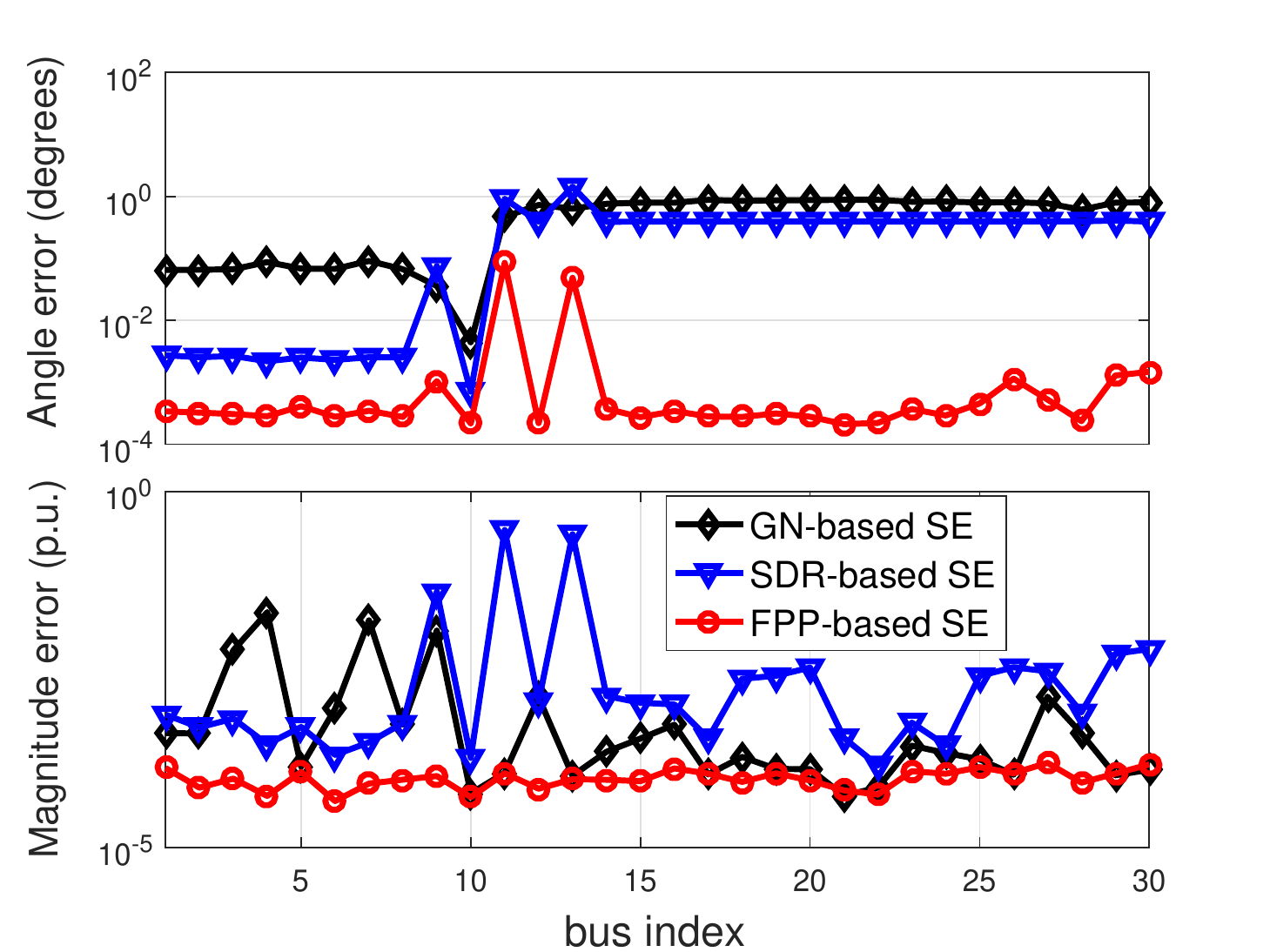}
\end{subfigure}
\begin{subfigure}
\centering
\includegraphics[width=66mm,scale=2]{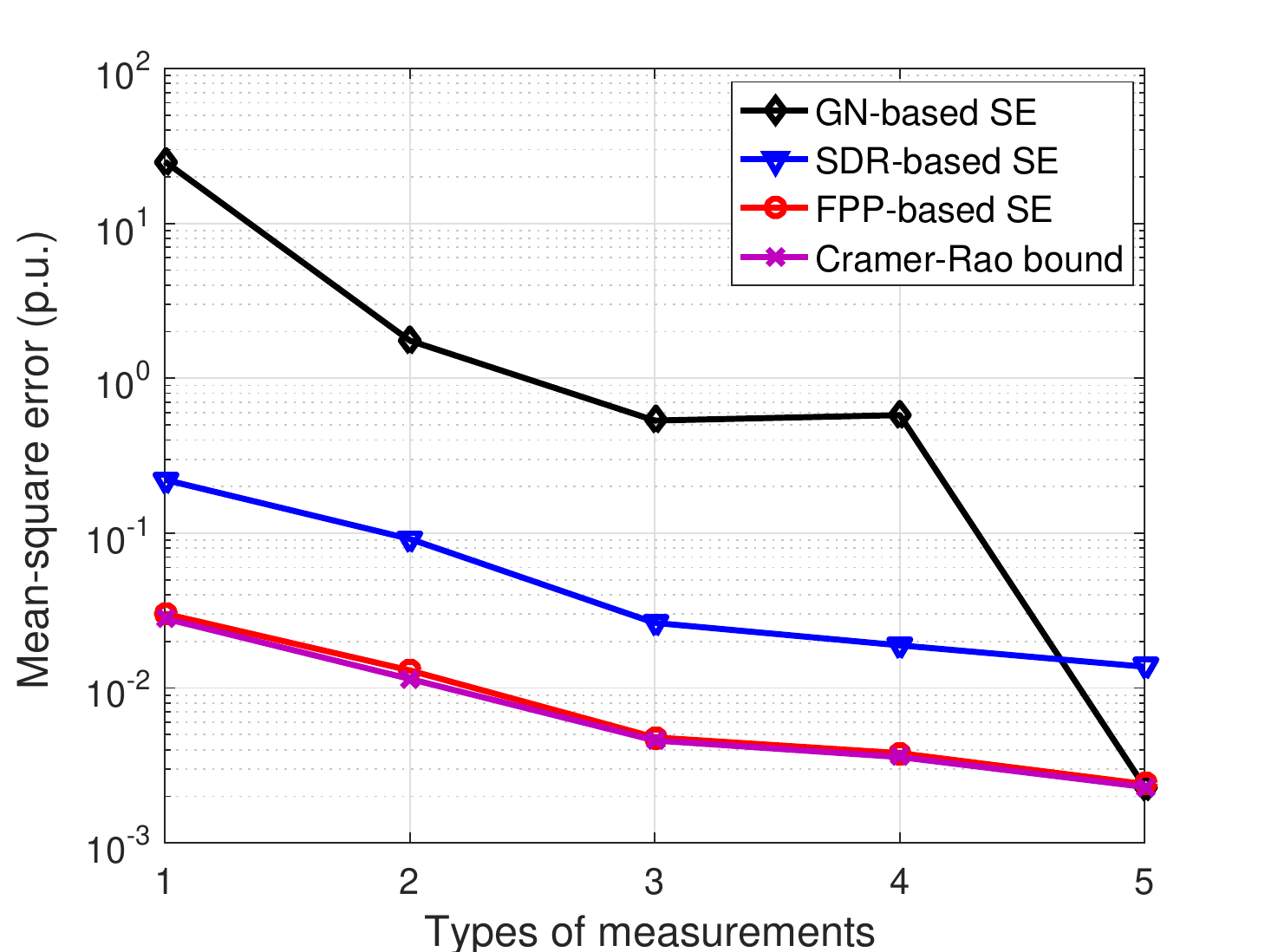}
\end{subfigure}
\hspace{-1em}
\caption{Left: Voltage magnitude and angle estimation errors per bus for the IEEE $30$-bus system. Right: MSEs and CRLB versus types of measurements used for the IEEE $14$-bus system using: (i) Gauss-Newton iterations; ii) the SDR-based PSSE; and iii) the FPP-based PSSE.}\label{fig:30bus}
\end{figure}

Figure~\ref{fig:30bus} compares Gauss-Newton iterations, the SDR-based solver, and the FPP-based solver on the IEEE 14- and 30-bus systems~\cite{PSTCA}. The actual nodal voltage magnitudes and angles were generated uniformly at random over $[0.9,1.1]$ and $[-0.4\pi,0.4\pi]$, respectively. Independent zero-mean Gaussian noise with standard deviation $0.05$ for power and $0.02$ for voltage measurements was assumed, and all reported results were averaged over $100$ independent Monte Carlo realizations. The measurements for the IEEE 30-bus system include all nodal voltage magnitudes and the active power flows at both sending and receiving ends. The left panel of Fig.~\ref{fig:30bus} depicts that the magnitude and angle estimation errors attained by the FPP solver are consistently below its competing alternatives \cite{tps2017wzgs}. 

The second experiment examines the MSE performance of the three approaches relative to the CRLB of \eqref{eq:crlb} for the IEEE $14$-bus test system. Initially, all voltage magnitudes as well as all sending- and receiving-end active power flows were measured, which corresponds to the base case 1 in the $x$-axis of the right panel of Fig.~\ref{fig:30bus}. To show the MSE performance relative for an increasing number of measurements, additional types of measurements were included in a deterministic manner. All types of SCADA measurements were ordered as $\{V_n^2,P_{kn},P_{nk},Q_{kn},Q_{nk},P_n,Q_n\}$. Each $x$-axis value in the right panel of Fig.~\ref{fig:30bus} implies that the number of ordered types of measurements was used in the experiment to obtain the corresponding MSEs.

\subsection{Synchrophasors}\label{subsec:problem:PMU}
To incorporate synchrophasors into the PSSE formulation, let $\bzeta_n = \bPhi_n \bv + \bvarepsilon_n$ collect the noisy PMU data at bus $n$ [cf.~\eqref{eq:WLS2.5}]. The related measurement matrix is $\bPhi_n$, and the measurement noise $\bvarepsilon_n$ is assumed to be complex zero-mean Gaussian, independent from the noise $\bepsilon$ in legacy meters and across buses. Following the normalization convention in \eqref{eq:WLS5}, the noise vector $\bvarepsilon_n$ is assumed prewhitened, such that all PMU measurements exhibit the same accuracy. The PSSE task now amounts to estimating $\bv$ given both $\bz$ and $\{\bzeta_n\}_{n\in\cP}$, where $\cP \subseteq \cB$ denotes the subset of the PMU-instrumented buses. Hence, the MLE cost in \eqref{eq:WLS5} needs to be augmented by the log-likelihood induced by PMU data as
\begin{equation}
 \hat{\bv} := \arg \min_\bv ~\sum_{m=1}^M 
(z_m - h_m (\bv))^2 + \sum_{n\in\cP} \| \bzeta_n -\bPhi_n\bv\|_2^2. \label{SEa_wls}
\end{equation}
The SDR methodology is again well motivated to convexify the augmented PSSE problem \eqref{SEa_wls} into
\begin{subequations}\label{SEa_sdp}
\begin{align}
\underset{\bV,\;\bv,\;\bchi}{\minimize} ~&~\bone^\cT\bchi + \sum_{n\in\cP}  \big[\trace(\bPhi_n^\cH \bPhi_n \bV) - 2 \Re\{\bzeta_n^\cH \bPhi_n\bv\}\big]  \label{SEa_sdpf}\\
\subjectto ~&~  \!
\left[\begin{array}{cc} \bV & \bv \\ \bv^\cH & 1 \end{array}\right] \succeq \bzero,~~\text{and}~~\eqref{SE_sdpocc}.\label{SEa_sdpc}
\end{align}
\end{subequations}
By Schur's complement, the left SDP constraint in \eqref{SEa_sdpc} can be expressed equivalently as $\bV\succeq \bv\bv^\cH$. If the latter constraint is enforced with equality, the matrix $\bV$ becomes rank-one. Imposing a rank-one constraint in \eqref{SEa_sdp} renders it equivalent to the augmented PSSE task of \eqref{SEa_wls}. The SDP here also offers the advantages of \eqref{SE_sdp}, in terms of the near-optimality and the distributed implementation deferred to Section~\ref{subsec:distsolv:sdr}. To recover a feasible solution, one can again use the best rank-1 approximation or adopt the randomization technique as elucidated in \cite{jstsp2014zhu}. 

Alternatively, the two types of measurements can be jointly utilized upon interpreting the SCADA-based estimate as a prior for PMU-based estimation~\cite{PhTh08,KeGiWo12}. Specifically, if $\hat{\bv}_s$ is the SCADA-based estimate, the prior probability density function of the actual state can be postulated to be a circularly symmetric complex Gaussian with mean $\hat{\bv}_s$ and covariance $\hat{\bSigma}_s$.

Given PMU data and the SCADA-based prior, the state can be estimated following a maximum \emph{a-posteriori} probability (MAP) approach as
\begin{equation}\label{eq:MAP}
 \hat{\bv} := \arg \min_\bv~(\bv-\hat{\bv}_s)^{\cH}\hat{\bSigma}_s^{-1}(\bv-\hat{\bv}_s) + \sum_{n\in\cP} \| \bzeta_n -\bPhi_n\bv\|_2^2
\end{equation}
where the first summand is the negative logarithm of the prior distribution, and the second one is the negative log-likelihood from the PMU data. In essence, the approach in \eqref{eq:MAP} treats the SCADA-based estimate as pseudo-measurements relying on the model $\hat{\bv}_s=\bv + \bepsilon$ with circularly symmetric zero-mean noise having $\mathbb{E}[\bepsilon\bepsilon^{\cH}]=\hat{\bSigma}_s$.

\section{Distributed Solvers}\label{sec:distsolv}
Upcoming power system requirements call for decentralized solvers. Measurements are now collected at much finer spatio-temporal scales and the number of states increases exponentially as monitoring schemes extend to low-voltage distribution grids~\cite{ReCeTh10}. Tightly interconnected power systems call for the close coordination of regional control centers~\cite{ExpositoPROC11}, while operators and utilities perform their computational operations \emph{on the cloud}.

This section reviews advances in distributed PSSE solvers. As the name suggests, distributed PSSE solutions spread the computational load across different processors or control centers to speed up time, implement memory-intensive tasks, and/or guarantee privacy. A network of processors may be coordinated by one or more supervising control centers in a hierarchical fashion, or completely autonomously, by exchanging information between processors. To clarify terminology, the latter architecture will be henceforth identified as decentralized.

Distributed solvers with a hierarchical structure have been proposed since the statistical formulation of PSSE~\cite[Part III]{Schweppe70}. Different versions of this original scheme were later developed in~\cite{cutsem83}, \cite{Iwamoto89}, \cite{ZhaoAbur05}, \cite{ExpositoPROC11}, \cite{Korres11}. Decentralized schemes include block Jacobi iterations \cite{LinLin94}, \cite{Conejo07}; an approximate algorithm building on the related optimality conditions~\cite{Falcao95}; or matrix-splitting techniques for facilitating matrix inversion across areas running Gauss-Newton iterations~\cite{tps2016mll}. Most of the aforementioned approaches presume local identifiability (i.e., each area is identifiable even when shared measurements are excluded) or their convergence is not guaranteed. Assuming a ring topology, every second agent updates its state iteratively through the auxiliary problem principle in~\cite{Ebrah00}. Local observability is waived in the consensus-type solver of \cite{XieChoiKar11}, where each control center maintains a copy of the entire high-dimensional state vector resulting in slow convergence. For a relatively recent review on distributed PSSE solves, see also~\cite{MASEsurvey}.

\subsection{Distributed Linear Estimators}\label{subsec:distsolv:linear}
Consider an interconnected system partitioned in $K$ areas supervised by separate control centers. Without loss of generality, an area may be thought of as an independent system operator region, a balancing authority, a power distribution center, or a substation~\cite{PhTh08}. Area $k$ collects $M_k$ measurements obeying the linear model
\begin{equation}\label{eq:DPSSE1}
\mathbf{z}_k=\bH_k\bv_k + \bepsilon_k
\end{equation}
where vector $\bv_k\in \mathbb{C}^{N_k}$ collects the system states related to $\bz_k$ through the complex matrix $\bH_k$. The random noise vector $\bepsilon_k$ is zero-mean with identity covariance upon prewhitening, if measurements are uncorrelated across areas. The model in \eqref{eq:DPSSE1} is exact for PMU measurements, but it may also correspond to a single Gauss-Newton iteration as explained in Section~\ref{subsec:problem:GN}.

\begin{figure}[t]
\begin{subfigure}
\centering
\includegraphics[width=0.48\linewidth]{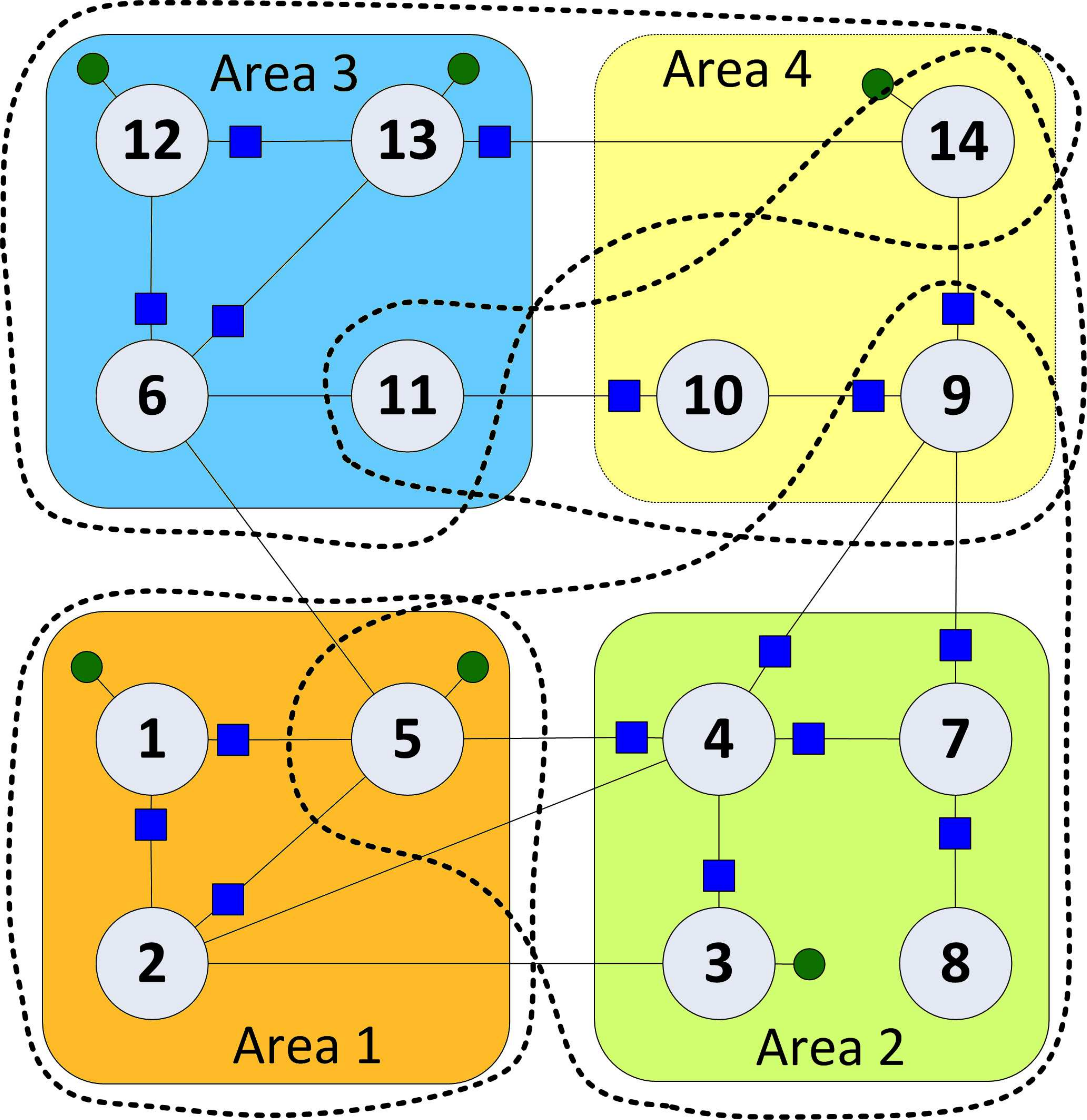}
\end{subfigure}
\hspace{2mm}
\begin{subfigure}
\centering
\includegraphics[width=0.47\linewidth]{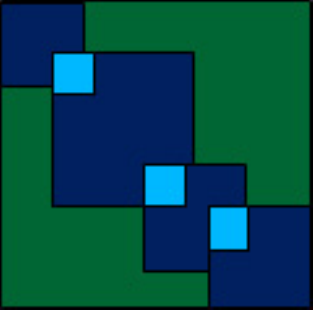}
\end{subfigure}
\caption{Left: The IEEE 14-bus system partitioned into four areas \cite{PSTCA,Korres11}. Dotted lassos show the buses belonging to area state vectors $\bv_k$'s. PMU bus voltage (line current) measurements depicted by green circles (blue squares). Right: The matrix structure for the left system with the distributed SDR solver: green square denotes the overall $\bV$, while dark and light blue ones correspond to the four area submatrices $\{\bV_k\}$ and their overlaps.}
\label{fig:ieee14}
\end{figure}

Performing PSSE locally at area $k$ amounts to solving
\begin{equation}\label{eq:local_problem}
\underset{\bv_k\in\cX_k} {\minimize}~f_k(\bv_k)
\end{equation}
where the convex set $\cX_k$ captures possible prior information, such as zero-injection buses or short circuits~\cite{Mo00}, \cite{AburExpositoBook}. If $f_k(\bv_k)=\|\bz_k-\bH_k\bv_k\|_2^2/2$, the minimizer of \eqref{eq:local_problem} is the least-squares estimate (LSE) of $\bv_k$, which is also the MLE of $\bv_k$ for Gaussian $\bepsilon_k$. 

As illustrated in Fig.~\ref{fig:ieee14}, the per-area state vectors $\{\bv_k\}_{k=1}^K$ overlap partially. Although area 4 supervises buses $\{9,10,14\}$, it also collects the current reading on lines $(10,11)$. Thus, its state vector $\bv_4$ extends to bus $\{11\}$ that is nominally supervised by area 3. To setup notation, define $\cS_{kl}$ as the shared states for a pair of neighboring areas $(k,l)$. Let also $\bv_{k}[l]$ ($\bv_{l}[k]$) denote the sub-vector of $\bv_k$ ($\bv_l$) consisting of their overlapping variables ordered as they appear in $\bv$. For example, $\bv_3[4]=\bv_4[3]$ contain the bus voltages of $\{11\}$. Solving the $K$ problems in \eqref{eq:local_problem} separately is apparently suboptimal since the estimates of shared states will disagree, tie-line measurements have to be ignored, and boundary states may thus become unobservable. 

Coupling the per-area PSSE tasks can be posed as
\begin{align}\label{eq:d_problem}
\underset{\{\bv_k\in \cX_k\}}{\minimize}~&~\!\sum_{k=1}^K f_k(\bv_k)\\
\subjectto~&~ \bv_{k}[l]=\bv_{l}[k],\quad \forall l\in\cB_k,~\forall k\nonumber
\end{align}
where $\cN_k$ is the set of areas sharing states with area $k$. The equality constraints of \eqref{eq:d_problem} guarantee consensus over the shared variables. Although \eqref{eq:d_problem} is amenable to decentralized implementations (cf.~\cite{Ebrah00}), areas need a coordination protocol for their updates. To enable a truly decentralized solution, we follow the seminal approach of~\cite{Sch08,Zhu09}. An auxiliary variable $\bv_{kl}$ is introduced per pair of connected areas $(k,l)$; the symbols $\bv_{kl}$ and $\bv_{lk}$ are used interchangeably. The optimization in \eqref{eq:d_problem} can then be written as
\begin{align}\label{eq:d_problem2}
\underset{\{\bv_k\in\mathcal{X}_k\},\;\{\bv_{kl}\}}{\minimize}~&~\sum_{k=1}^K f_k(\bv_k)\\
\subjectto ~&~ \bv_{k}[l]=\bv_{kl},\quad \forall l\in \cB_k,~k=1,\ldots,K.\nonumber
\end{align}

Problem \eqref{eq:d_problem2} can be solved using the alternating direction method of multipliers (ADMM)~\cite{Boyd10}, \cite{chap2015gg}. In its general form, ADMM tackles convex optimization problems of the form
\begin{subequations}
\begin{align}\label{eq:ADMM}
\underset{\bx\in\cX,\;\bz\in\cZ}{\minimize}~&~f(\bx)+g(\bz)\\
\subjectto~&~\bA\bx+\bB\bz=\bc
\end{align}
\end{subequations}
for given matrices and vectors $(\bA,\bB,\bc)$ of proper dimensions. Upon assigning a Lagrange multiplier $\blambda$ for the coupling constraint in \eqref{eq:ADMM}, the $(\bx,\bz)$ minimizing \eqref{eq:ADMM} are found through the next iterations for some $\mu>0$
\begin{subequations}\label{eq:ADMM2}
\begin{align}
\bx^{i+1}&:=\arg\min_{\bx\in\cX}~f(\bx)+\frac{\mu}{2}\|\bA\bx+\bB\bz^{i}-\bc+\blambda^i\|_2^2\label{eq:ADMM2:a}\\
\bz^{i+1}&:=\arg\min_{\bz\in\cZ}~g(\bz)+\frac{\mu}{2}\|\bA\bx^{i+1}+\bB\bz-\bc+\blambda^i\|_2^2\label{eq:ADMM2:b}\\
\blambda^{i+1}&:=\blambda^{i} + \bA\bx^{i+1}+\bB\bz^{i+1}-\bc.\label{eq:ADMM2:c}
\end{align}
\end{subequations}

Towards applying the ADMM iterations to \eqref{eq:d_problem}, identify variables $\{\bv_k\}$ as $\bx$ in \eqref{eq:ADMM} and $\{\bv_{kl}\}$ as $\bz$ with $g(\bz)=0$. Moreover, introduce Lagrange multipliers $\blambda_{k,l}$ for each constraint in \eqref{eq:d_problem2}. Observe that $\blambda_{k,l}$ and $\blambda_{l,k}$ correspond to the distinct constraints $\bv_{k}[l]=\bv_{kl}$ and $\bv_{l}[k]=\bv_{kl}$, respectively. According to \eqref{eq:ADMM2:a}, the per-area state vectors $\{\bv_k\}$ can be updated separately as
\begin{equation}\label{eq:ADMM3}
\bv_k^{i+1}:=\arg\min_{\bv_k\in\cX_k}~f_k(\bv_k)+\frac{\mu}{2}\sum_{l\in\cB_k}\|\bv_k[l]-\bv_{kl}^i+\blambda_{k,l}^i\|_2^2.
\end{equation}

From \eqref{eq:ADMM2:b} and assuming every state is shared by at most two areas, the auxiliary variables $\bv_{kl}$ can be readily found in closed form given by
\begin{equation}\label{eq:ADMM4}
\bv_{kl}^{i+1}=\frac{1}{2}\left( \bv_k^{i+1}[l]+ \bv_l^{i+1}[k] + \blambda_{k,l}^i + \blambda_{l,k}^i\right)
\end{equation}
while the two related multipliers are updated as
\begin{subequations}\label{eq:ADMM5}
\begin{align}
\blambda_{k,l}^{i+1}&:=\blambda_{k,l}^{i} + (\bv_k^{i+1}[l] - \bv_{kl}^{i+1})\\
\blambda_{l,k}^{i+1}&:=\blambda_{l,k}^{i} + (\bv_l^{i+1}[k] - \bv_{kl}^{i+1}).
\end{align}
\end{subequations}
Adding \eqref{eq:ADMM5} by parts and combining it with \eqref{eq:ADMM4} yields $\blambda_{k,l}^{i+1}=-\blambda_{l,k}^{i+1}$ at all iterations $i$ if the multipliers are initialized at zero. Hence, the auxiliary variable $\bv_{kl}$ ends up being the average of the shared states; that is,
\begin{equation}\label{eq:ADMM6}
\bv_{kl}^{i+1}=\frac{1}{2}\left( \bv_k^{i+1}[l]+ \bv_l^{i+1}[k] \right).
\end{equation}

To summarize, at every iteration $i$:
\begin{enumerate}
\item[\bf (i)] Each control area solves \eqref{eq:ADMM3}. If $f_k(\bv_k)$ is the LS fit and for unconstrained problems, the per-area states are updated as the LSEs using legacy software. The second summand in \eqref{eq:ADMM3} can be interpreted as pseudo-measurements on the shared states forcing them to consent across areas.
\item[\bf (ii)] Neighboring areas exchange their updated shared states. This step involves minimal communication, and no grid models need to be shared. Every area updates its copies of the auxiliary variables $\bv_{kl}$ using \eqref{eq:ADMM6}.
\item[\bf (iii)] Every area updates the Lagrange multipliers $\blambda_{kl}$ based on the deviation of the local from the auxiliary variable as in \eqref{eq:ADMM5}.
\end{enumerate}

For convex pairs $\{f_k(\bv_k),\mathcal{X}_k\}_{k=1}^K$, the aforementioned iterates reach the optimal cost in \eqref{eq:d_problem}, under mild conditions. If the overall power system is observable, the ADMM iterates converge to the unique LSE. The approach has been extended in~\cite{CAMSAP13} for joint PSSE and breaker status verification.

\begin{figure}
\centering
\includegraphics[width=0.7\linewidth]{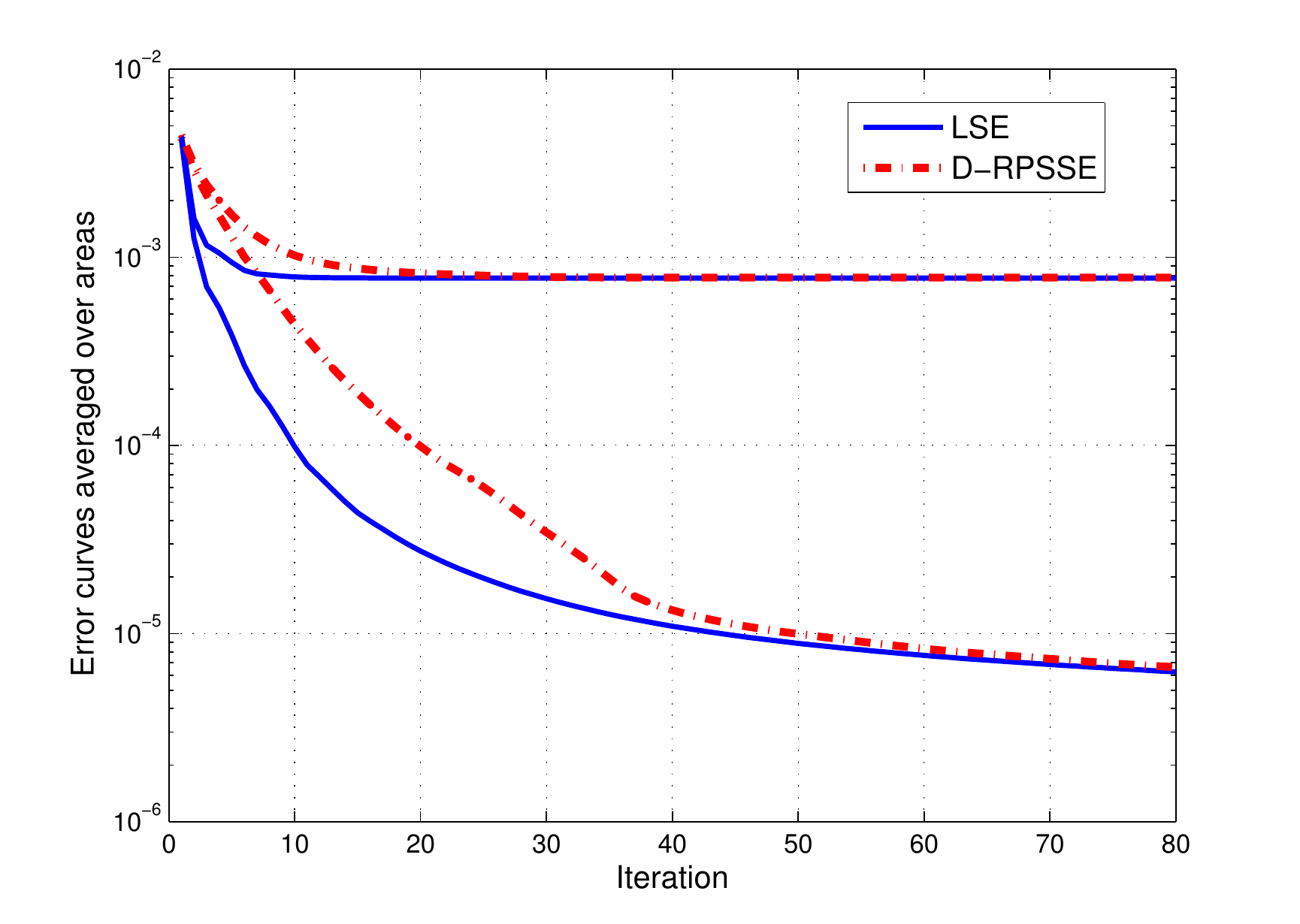}
\caption{Average error curves $\sum_{k=1}^{300}e_{k,c}^t/300$ (bottom) and $\sum_{k=1}^{300}e_{k,o}^t/300$ (top) for the LSE, and its robust counterpart D-RPSSE (see Section~\ref{subsec:robust:tests}) on a 4,200-bus grid.}
\label{fig:14x300}
\end{figure}

The decentralized algorithms were tested on a 4,200-bus power grid synthetically built from the IEEE 14- and 300-bus systems. Each of the 300 buses of the latter was assumed to be a different area, and was replaced by a copy of the IEEE 14-bus grid. Additionally, every branch of the IEEE 300-bus grid was an inter-area line whose terminal buses are randomly selected from the two incident to this line areas. Two performance metrics were adopted: the per area error to the centralized solution of \eqref{eq:d_problem}, denoted by $e_{k,c}^t{:=}\|\bv_k^{(c)}-\bv_k^t\|_2/N_k$, and the per-area error to the true underlying state defined as $e_{k,o}^t{:=}\|\bv_k-\bv_k^t\|_2/N_k$. Figure~\ref{fig:14x300} shows the corresponding error curves averaged over 300 areas. The decentralized LSE approached the underlying state at an accuracy of $10^{-3}$ in approximately 10 iterations or 6.2~msec on an Intel Duo Core @ 2.2 GHz (4GB RAM) computer using MATLAB; while the centralized LSE finished in 93.4~msec.

\subsection{Distributed SDR-based Estimators}\label{subsec:distsolv:sdr}
Although the SDR-PSSE approach incurs polynomial complexity when implemented as a convex SDP, its worst-case complexity is still $\ccalO(M^4 \sqrt{N_b} \log(1/\epsilon))$ for a given solution accuracy $\epsilon>0$ \cite{LuMa10}. For typical power networks, the number of measurements $M$ is on the order of the number of buses $N_b$, and thus the worst-case complexity becomes $\ccalO(N_b^{4.5}\log(1/\epsilon))$. This complexity could be prohibitive for large-scale power systems, which motivates accelerating the SDR-PSSE method using distributed parallel implementations.

Following the area partition in Fig.~\ref{fig:ieee14}, the $m$-th measurement per area $k$ can be written as 
\begin{align*}
z_{k,m} = h_{k,m}(\bv_{k})+ \epsilon_{k,m} = \trace(\bH_{k,m}
\bV_{k}) + \epsilon_{k,m}, ~\forall k , m 
\end{align*}
where $\bV_{k}$ denotes a submatrix of $\bV$ formed by extracting the rows and columns corresponding to buses in area $k$; and likewise for each $\bH_{k,m}$. Due to the overlap among the subsets of buses, the outer-product $\bV_{k}$ of area $k$ overlaps also with $\bV_{l}$ for each neighboring area $l\in \cN_k$, as shown in Fig.~\ref{fig:ieee14}.

By reducing the measurements at area $k$ to submatrix $\bV_k$, one can define the PSSE error cost $f_k(\bV_{k}) : = \sum_{m=1}^{M_k} \left[z_{k,m} - \trace(\bH_{k,m} \bV_{k}) \right]^2$ per area $k$, which only involves the local matrix $\bV_{k}$. Hence, the centralized PSSE problem in \eqref{SE_sdp} becomes equivalent to
\begin{equation}
\hat{\bV} = \arg \min_{\bV \succeq \bzero} ~ \sum_{k=1}^K f_k(\bV_{k}). \label{cse}
 \end{equation}
This equivalent formulation effectively expresses the overall PSSE cost as the superposition of each local cost $f_k$. Nonetheless, even with such a decomposable cost, the main challenge to implement \eqref{cse} in a distributed manner lies in the PSD constraint that couples the overlapping local matrices $\{\bV_{k}\}$ (cf. Fig. \ref{fig:ieee14}). If all submatrices $\{\bV_{k}\}$ were non-overlapping, the cost would be decomposable as in \eqref{cse}, and the PSD of $\bV$ would boil down to a PSD constraint per area $k$, as in
\begin{equation}\label{cse2}
\hat{\bV} = \arg \min_{\{\bV_{k} \succeq \bzero\}} ~\!\sum_{k=1}^K f_k(\bV_{k}).
\end{equation}
Similar to PSSE for linearized measurements in \eqref{eq:d_problem}, the formulation in \eqref{cse2} can be decomposed into sub-problems,  thanks to the separable PSD constraints. It is not always equivalent to the centralized \eqref{cse} though, because the PSD property of all submatrices does not necessarily lead to a PSD overall matrix. Nonetheless, the decomposable problem \eqref{cse2} is still a valid SDR-PSSE reformulation, since with the additional per-area constraints $\rank(\bV_{k})=1$, it is actually equivalent to \eqref{SE_sdpo}. While it is totally legitimate to use \eqref{cse2} as the relaxed SDP formulation for \eqref{SE_sdpo}, the two relaxed problems are actually equivalent under mild conditions.

The fresh idea here is to explore valid network topologies to facilitate such PSD constraint decomposition. To this end, it will be instrumental to leverage results on completing partial Hermitian matrices to obtain PSD ones~\cite{grone_psd84}. Upon obtaining the underlying graph formed by the specified entries in the partial Hermitian matrices, these results rely on the so-termed graph \emph{chordal} property to establish the equivalence between the positive semidefiniteness of the overall matrix and that of all submatrices corresponding to the graph's maximal cliques. Interestingly, this technique was recently used for developing distributed SDP-based optimal power flow (OPF) solvers in \cite{jabr_tps12,edhzgg_tsg13,lam_dopt11}. 

Construct first a new graph $\cB'$ over $\cB$, with all its edges corresponding to the entries in $\{\bV_{k}\}$. The graph $\cG'$ amounts to having all buses within each subset $\cN_{k}$ to form a clique. Furthermore, the following are assumed:

{\it
\begin{enumerate}
\itemsep -2pt
\item
[$(\mathbf{as4})$] The graph with all the control areas as nodes, and
their edges defined by the neighborhood subset $\{\cN_k\}_{k=1}^K$ forms a tree.
\item
[$(\mathbf{as5})$] Each control area has at least one bus that does
not overlap with any neighboring area.
\end{enumerate}
}
\begin{proposition}
\label{prop:equivalence} Under (as4)-(as5), the two relaxed problems \eqref{cse} and \eqref{cse2} are equivalent.
\end{proposition}

\begin{figure}[t]
	\begin{subfigure}
		\centering
		\includegraphics[width=66mm,scale=2]{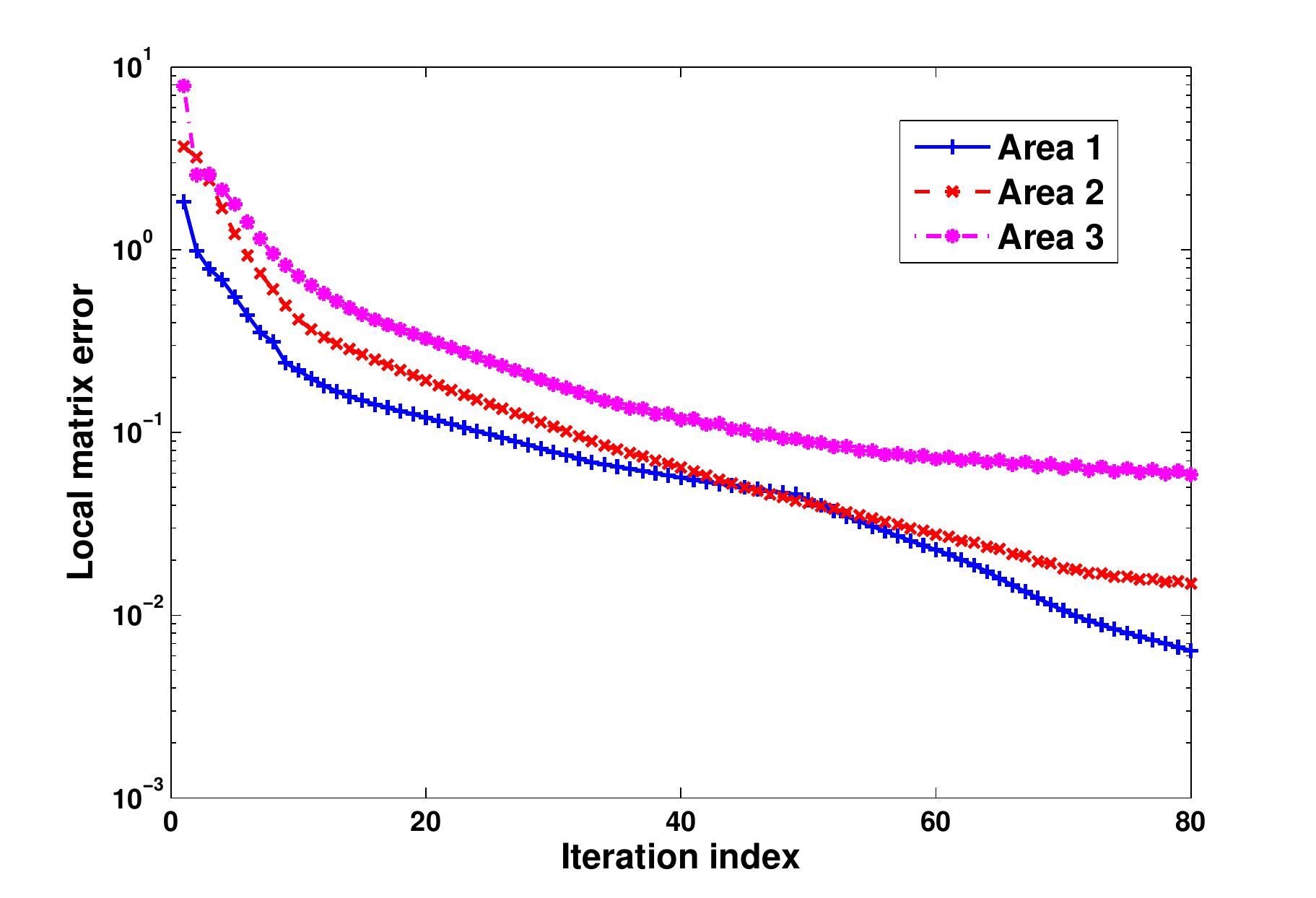}
	\end{subfigure}
	\begin{subfigure}
		\centering
		\includegraphics[width=64mm,scale=2]{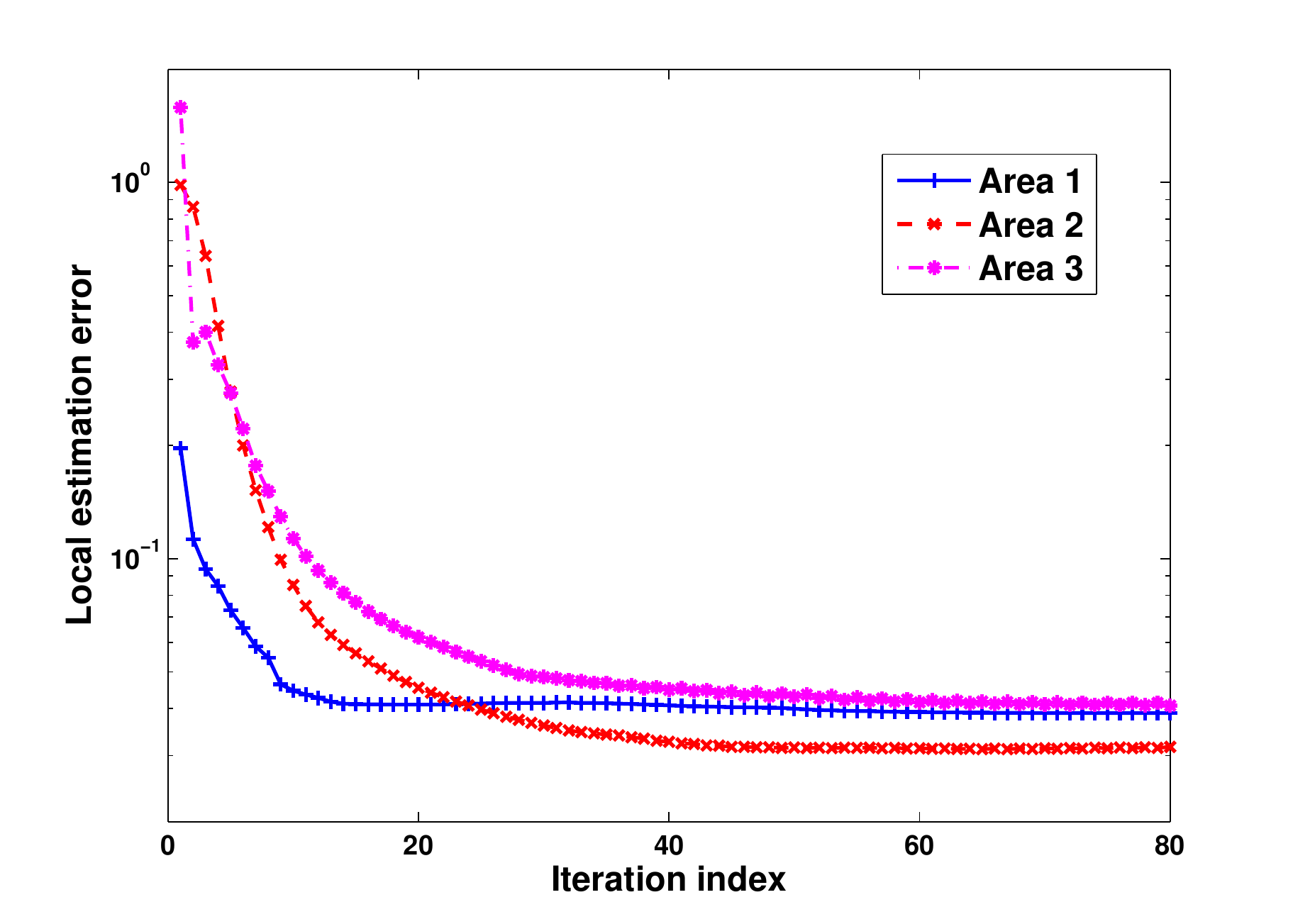}
	\end{subfigure}
	\hspace{-1em}
	\caption{(Left) Per area state matrix  error and (Right) state vector estimation error, versus the number of ADMM iterations for the distributed SDR-PSSE solver using the IEEE 118-bus system.
	}\label{fig:dsdr}
\end{figure}

Proposition~\ref{prop:equivalence} can be proved by following the arguments in \cite{jstsp2014zhu} to show that the entire PSD matrix $\bV$ can be ``completed'' using only the PSD submatrices $\bV_{k}$. The key point is that in most power networks even those not obeying {\bf (as4)} and {\bf (as5)}, \eqref{cse2} can achieve the same accuracy as the centralized one. At the same time, decomposing the PSD constraint in \eqref{cse2} is of paramount importance for developing distributed solvers. One can adopt the consensus reformulation to design the distributed solver for \eqref{cse2} as in \eqref{eq:d_problem} of Section~\ref{subsec:distsolv:linear}. Accordingly, the ADMM iterations can be employed to solve \eqref{cse2} through iterative information exchanges among neighboring areas, and this is the basis of the distributed SDR-PSSE method.

This distributed SDR-PSSE method was tested on the IEEE 118-bus system using the three-area partition in \cite{KeGi12}. All three areas measure their local bus voltage magnitudes, as well as real and reactive power flow levels at all lines. The overlaps among the three areas form a tree communication graph used to construct the area-coupling constraints. To demonstrate convergence of the ADMM iterations to the centralized
SE solution $\hat{\bV}$ of \eqref{cse2}, the local matrix Frobenius error norm $\|\hat{\bV}_{(k)}^i-\hat{\bV}_{(k)}\|_F$ is plotted versus the iteration index $i$ in the left panel of Fig. \ref{fig:dsdr} for every control area $k$.
Clearly, all local iterates converge to (approximately with a linear rate) their counterparts in the centralized solution. As the task of interest is to estimate the voltages, the local estimation error for the state vector
$\|\hat{\bv}^i_{(k)} - \bv_{(k)}\|_2$ is also depicted in the right panel of Fig.
\ref{fig:dsdr}, where $\hat{\bv}_{(k)}$ is the estimate of bus voltages
at area $k$ obtained from the iterate $\hat{\bV}_{(k)}^i$ using the
eigen-decomposition method. Interestingly, the estimation error
costs converge within the estimation accuracy of around $10^{-2}$ after about 20 iterations (less than 10 iterations for area 1),
even though the local matrix has not yet converged. In addition, these error costs decrease much faster in the first couple of iterations. This demonstrates that even with only a limited number of iterations, the PSSE accuracy can be greatly boosted in practice, which in turn makes inter-area communication overhead more affordable.

\section{Robust Estimators and Cyber Attacks}\label{sec:robust}
Bad data, also known as \emph{outliers} in the statistics parlance, can challenge PSSE due to communication delays, instrument mis-calibration, and/or line parameter uncertainty. In today's cyber-enabled power systems, smart meter and synchrophasor data could be also purposefully manipulated to mislead system operators. This section reviews conventional and contemporary approaches to coping with outliers.

\subsection{Bad Data Detection and Identification}\label{subsec:robust:tests}
Bad data processing in PSSE relies mainly on the linear measurement model $\bz=\bH\bv+\bepsilon$, where $\bH\in\mathbb{R}^{M\times N}$. Recall that this model is exact for PMU measurements~[cf.~\eqref{eq:WLS2.5}--\eqref{eq:WLS2.6}], but approximate per Gauss-Newton iteration or under the linearized grid model. In addition, the aforementioned model assumes real-valued states and measurements, slightly abusing the symbols introduced in \eqref{eq:WLS2.6}. This is to keep the notation uncluttered and cover both cases of exact and inexact grid models. Albeit the nominal measurement noise vector is henceforth assumed zero-mean with identity covariance, results extend to colored noise as per \eqref{eq:prewhitened}.

To capture bad data, the measurement model is now augmented as
\begin{equation}\label{eq:model+o}
\bz=\bH\bv + \bo+\bepsilon
\end{equation}
where $\bo\in\mathbb{R}^M$ is an unknown vector whose $m$-th entry $o_m$ is deterministically non-zero only if $z_m$ is a bad datum \cite{VKGG11,Kosut11,DuanYangScharf11}. Therefore, vector $\bo$ is \emph{sparse}, i.e., many of its entries are zero. Under this outlier-cognizant model in \eqref{eq:model+o}, the unconstrained LSE given as $\hat{\bv}_\text{LSE}=(\bH^{\cT}\bH)^{-1}\bH^{\cT}\bz$, yields the residual error
\begin{equation}\label{eq:residual}
\br:=\bz-\bH\hat{\bv}_\text{LSE}=\bP\bz=\bP(\bo+ \bepsilon)
\end{equation}
with $\bP:=\bI_M - \bH(\bH^{\cT}\bH)^{-1}\bH^{\cT}$ being the so-called projection matrix onto the orthogonal subspace of $\range(\bH)$. The last equality in \eqref{eq:residual} stems from the fact that $\bP\bH=\bzero$. As a projection matrix, $\bP$ is idempotent, that is, $\mathbf{P}=\mathbf{P}^2$; Hermitian PSD with $(M-N)$ eigenvalues equal to one and $N$ zero eigenvalues; while its diagonal entries satisfy $\bP_{m,m}\in[0,1]$ for $m=1,\,\ldots,\,M$; see e.g., \cite{CelikAbur92}. 

For $\bepsilon\sim \mathcal{N}(\bzero,\bI_M)$, it apparently holds that $\bP\bepsilon\sim\cN(\bzero,\bP)$. The mean-squared residual error is (see also \cite{Kosut11} for its Bayesian counterpart)
\begin{equation}\label{eq:mlse}
\mathbb{E}[\|\br\|_2^2]=\mathbb{E}[\|\bP\bepsilon\|_2^2] + \|\bP\bo\|_2^2=(M-N)+\|\bP\bo\|_2^2.
\end{equation}
In the absence of bad data, or if $\bo\in \range(\bH)$, the squared residual error follows a $\chi^2$ distribution with mean $(M-N)$. The $\chi^2$-test compares $\|\br\|_2^2$ against a threshold to detect the presence of bad data \cite{AburExpositoBook,Mo00}. 

Finding both $\bv$ and $\bo$ from measurements in \eqref{eq:model+o} may seem impossible, given that the number of unknowns exceeds the number of equations. Leveraging the sparsity of $\mathbf{o}$ though, interesting results can be obtained \cite{VKGG11}. If $\tau_0$ bad data are expected, one would ideally wish to solve
\begin{equation}\label{eq:ell_0}
\{\hat{\bv},\hat{\bo}\}\in\arg\min_{\bv,\;\bo}~\Big\{ \frac{1}{2}\left\|\bz-\bH\bv-\bo\right\|_2^2:~\|\bo\|_0\leq \tau_0\Big\}.
\end{equation}
But the $\ell_0$-(pseudo) norm $\|\bo\|_0$ counting the number of non-zero entries of $\bo$, renders \eqref{eq:ell_0} NP-hard in general; see also Definition~\ref{def:id} later in Section \ref{subsec:robust:attacks}.

For the special case of $\tau_0=1$, problem \eqref{eq:ell_0} can be efficiently handled. Consider the scenario where the only non-zero entry of $\hbo$ is the $m$-th one, and denote the related $\hbv$ minimizer by $\hbv_{(m)}$. Apparently, the $m$-th entry of the $\hbo$ minimizer is $\hat{o}_m:=z_m-\bh_m^{\cT}\hbv_{(m)}$. This choice nulls the $m$-th residual $(z_m-\bh_m^{\cT}\hbv_{(m)}-\hat{o}_m=0)$. With the $m$-th residual zeroed, the cost in \eqref{eq:ell_0} becomes $\|\br_{(m)}\|_2^2:=\|\bz_{(m)}-\bH_{(m)}\hbv_{(m)}\|_2^2$, where $\bz_{(m)}$ is obtained from $\bz$ upon dropping its $m$-th entry and $\bH_{(m)}$ by removing the $m$-th row of $\bH$. The problem in \eqref{eq:ell_0} is then equivalent to
\begin{equation}\label{eq:ell_01}
\underset{m}{\minimize}~\frac{1}{2}\|\br_{(m)}\|_2^2.
\end{equation}
Problem \eqref{eq:ell_01} can be solved by exhaustively finding all $M$ LSEs excluding one measurement at a time. Fortunately, a classical result from the adaptive filtering literature relates the error $\|\br_{(m)}\|_2^2$ to the error attained using all outlier-free measurements $\|\br\|_2^2:=\|\bP\bz\|_2^2$; see e.g.,~\cite[Ch.~9]{Haykin}
\begin{equation}\label{eq:MSEs}
\|\br\|_2^2=\|\br_{(m)}\|_2^2+r_m \hat{o}_m.
\end{equation}
The same result links the \emph{a-posteriori} error $r_m$ to the \emph{a-priori} error $\hat{o}_m$ as $r_m=\bP_{m,m} \hat{o}_m$. Through these links, solving \eqref{eq:ell_01} is equivalent to: 
\begin{equation}\label{eq:LNR}
r_{\max}:=\underset{m}{\maximize}~\frac{|r_m|}{\sqrt{\bP_{m,m}}}.
\end{equation}
In words, a single bad datum can be identified by properly normalizing the entries of the original residual vector $\br=\bP\bz$.

Interestingly, the task in \eqref{eq:LNR} coincides with the largest normalized residual (LNR) test that compares $r_{\max}$ to a prescribed threshold to identify a single bad datum~\cite[Sec.~5.7]{AburExpositoBook}. The threshold is derived after recognizing that in the absence of bad data, $r_m/\sqrt{\bP_{m,m}}$ is standard normal for all $m$.

The LNR test does not generalize for multiple bad data and problem \eqref{eq:ell_0} becomes computationally intractable for larger $\tau_0$'s. Heuristically, if a measurement is deemed as outlying, PSSE is repeated after discarding this bad datum, the LNR test is re-applied, and the process iterates till no corrupted data are identified. Alternatively, the \emph{least-median squares} and the \emph{least-trimmed squares} estimators have provable breakdown points and superior efficiency under Gaussian data; see e.g., \cite{Mili94} and references therein. Nevertheless, their complexity scales unfavorably with the network size.

Leveraging compressed sensing~\cite{ChDoSa98}, a practical robust estimator can be found if the $\ell_0$-pseudonorm is surrogated by the convex $\ell_1$-norm as~\cite{VKGG11,KeGi12}
\begin{align}\label{eq:ell_1_c}
\underset{\bv,\;\bo}{\minimize}~\Big\{ \frac{1}{2}\left\|\bz-\bH\bv-\bo\right\|_2^2:~\|\bo\|_1\leq \tau_1\Big\}
\end{align}
for a preselected constant $\tau_1>0$, or in its Lagrangian form
\begin{align}\label{eq:ell_1}
\{\hbv,\hbo\}\in\arg\min_{\bv,\;\bo}~ \frac{1}{2}\left\|\bz-\bH\bv-\bo\right\|_2^2 + \lambda\|\bo\|_1
\end{align}
for some tradeoff parameter $\lambda>0$. The estimates of \eqref{eq:ell_1} offer joint state estimation and bad data identification. Even when some measurements are deemed as corrupted, their effect has been already suppressed. The optimization task in \eqref{eq:ell_1} can be handled by off-the-shelf software or solvers customized to the compressed sensing setup. When $\lambda\rightarrow \infty$, the minimizer $\hbo$ becomes zero, and thus $\hbv$ reduces to the LSE. On the contrary, by letting $\lambda \rightarrow 0^{+}$, the solution $\hbv$ coincides with the \emph{least-absolute value} (LAV) estimator~\cite{1982baddata,CelikAbur92,ElKeib92,2017wgc}; presented earlier in \eqref{eq:pPSSE1}, namely
\begin{equation}\label{eqLLAV}
\hat{\bv}_{\rm LAV}:=\arg\min_{\bv}~\|\bz-\bH\bv\|_1.
\end{equation}

For finite $\lambda>0$, the $\hbv$ minimizer of \eqref{eq:ell_1} is equivalent to Huber's M-estimator; see \cite{VKGG11} and references therein. Based on this connection and for Gaussian $\bepsilon$, parameter $\lambda$ can be set to 1.34, which makes the estimator 95\% asymptotically efficient for outlier-free measurements~\cite[p.~26]{MaMaYo06}. Huber's estimate can be alternatively expressed as the $\bv$-minimizer of \cite{MaMu00}, $\minimize_{\bv,\,\bomega} \,\tfrac{1}{2}\|\bomega\|_2^2 + \lambda \|\bv-\bH\bv-\bomega\|_1$. The bad data identification performance of this minimization has been analyzed in~\cite{tsp2013xumeng}.

Table~\ref{tbl:outliers} compares several bad data analysis methods on the IEEE 14-bus grid of Fig.~\ref{fig:ieee14} under the next four scenarios: \textbf{(S0)} no bad data; \textbf{(S1)} bad data on line $(4,7)$; \textbf{(S2)} bad data on line current $(4,7)$ and bus voltage $5$; and \textbf{(S3)} bad data on bus voltage $5$ and line currents $(4,7)$ and $(10,11)$. In all scenarios, bad data are simulated by multiplying the real and imaginary parts of the actual measurement by $1.2$. The performance metric here is the $\ell_2$-norm between the true state and the PSSE, which is averaged over 1,000 Monte Carlo runs. Four algorithms were tested: (a) an ideal but practically infeasible genie-aided LSE (GA-LSE), which ignores the corrupted measurements; (b) the regular LSE; (c) the LNR test-based (LNRT) estimator with the test threshold set to 3.0 \cite{AburExpositoBook}; and (d) Huber's estimator of \eqref{eq:ell_1} with $\lambda=1.34$. For (S0)-(S1), the estimators perform comparably. The few corrupted measurements in (S2)-(S3) can deteriorate LSE's performance, while Huber's estimator performs slightly better than LNRT. Computationally, Huber's estimator was run within 1.3~msec, while the LNRT required 1.5~msec. The computing times were also measured for the IEEE 118-bus grid without corrupted data. Interestingly, the average time on the IEEE 118-bus grid without corrupted data are 3.2~msec and 81~msec, respectively.

\begin{table}
\renewcommand{\arraystretch}{1.2}
\caption{Mean-Square Estimation Error in the Presence of Bad Data}
\vspace*{.5em}
\centering
\begin{tabular}{|c|r|r|r|r|r|}
\hline
\textbf{Method} & GA-LSE & LSE & LNRT & Huber's \\
\hline
(S0)			& $0.0278$	& $0.0278$	& $0.0286$	& $0.0281$\\
(S1)			& $0.0313$	& $0.0318$	& $0.0331$	& $0.0322$\\
(S2)			& $0.0336$	& $0.1431$	& $0.0404$	& $0.0390$\\
(S3)			& $0.0367$	& $0.1434$	& $0.0407$	& $0.0390$\\
\hline
\end{tabular}
\label{tbl:outliers}
\end{table}

Towards a robust decentralized state estimator, the ADMM-based framework of Section~\ref{subsec:distsolv:linear} can be engaged here too. If the measurement model for the $k$-th area is $\bz_k=\bH_k\bv_k+\bo_k+\bepsilon_k$, the centralized problem boils down to
\begin{align}\label{eq:Drobust}
\underset{\{\bv_k\in\mathcal{X}_k,\;\bo_k\}}{\minimize}~ \sum_{k=1}^K\frac{1}{2}\left\|\bz_k-\bH_k\bv_k-\bo_k\right\|_2^2 + \lambda\|\bo_k\|_1.
\end{align}
To allow for decentralized implementation, the optimization in \eqref{eq:Drobust} can be reformulated as
\begin{subequations}\label{eq:Drobust2}
\begin{align}
\minimize~&~\sum_{k=1}^K \frac{1}{2}\left\|\bz_k-\bH_k\bv_k-\bo_k\right\|_2^2 + \lambda\|\bomega_k\|_1\label{eq:Drobust2:cost}\\
\textrm{over}~&~\{\bv_k\in\mathcal{X}_k,\;\bo_k,\;\bomega_k\},\{\bv_{kl}\}\label{eq:Drobust2:vars}\\
\subjectto ~&~ \bv_{k}[l]=\bv_{kl},~\textrm{for all}~l\in \cB_k,~k=1,\ldots,K.\label{eq:Drobust2:con1}\\
~&~ \bo_k=\bomega_k,~\textrm{for all}~k=1,\ldots,K.\label{eq:Drobust2:con2}
\end{align}
\end{subequations}
As in Section~\ref{subsec:distsolv:linear}, the constraints in \eqref{eq:Drobust2:con1} and the auxiliary variables $\{\bv_{kl}\}$ enforce consensus of shared states. On the other hand, the variables $\{\bo_k\}$ are duplicated as $\{\bomega_k\}$ in \eqref{eq:Drobust2:con2}. Then, variables $\{\bv_k,\,\bo_k\}$ are put together in the $x$-update of ADMM in \eqref{eq:ADMM2:a}, whereas $\{\bv_{kl},\,\bomega_k\}$ fall into the $z$-update in \eqref{eq:ADMM2:b}. In this fashion, costs are separable over variable groups, and the minimization involving the $\ell_1$-norm enjoys a closed-form solution expressed in terms of the soft thresholding operator~\cite{KeGi12}.

\subsection{Observability and Cyber Attacks}\label{subsec:robust:attacks}
In the cyber-physical smart grid context, bad data are not simply unintentional errors, but can also take the form of malicious data injections~\cite{SecurityPROC12}. Amid these challenges, the intertwined issues of critical measurements and stealth cyber-attacks on PSSE are discussed next.

It has been tacitly assumed so far that the power system is \emph{observable}. A power system is observable if distinct states $\bv\neq \bv'$ are mapped to distinct measurements $\bh(\bv)\neq\bh(\bv')$ under a noiseless setup. Equivalently, if the so-called \emph{measurement distance function} is defined as~\cite{robust2012zhu}
\begin{equation}\label{eq:mdf}
D(\bh):=\underset{\bv\neq \bv'}{\minimize}~\|\bh(\bv)-\bh(\bv')\|_0
\end{equation}
the power system is observable if and only if $D(\bh)\geq 1$. Given the network topology and the mapping $\bh(\bv)$, the well-studied topic of observability analysis aims at determining whether the system state is uniquely identifiable, at least locally in a neighborhood of the current estimate~\cite[Ch.~4]{AburExpositoBook}. If not, mapping \emph{observable islands}, meaning maximally connected sub-grids with observable internal flows, is important as well. 

Observability analysis relies on the decoupled linearized grid model, and is accomplished through topological or numerical tests~\cite{ClKrDa83,MoWu85A}. Apparently, under the linear or linearized model $\bh(\bv)=\bH\bv$, the state $\bv$ is uniquely identifiable if and only if $\bH$ is full column-rank. Phase shift ambiguities can be waived by fixing the angle at a reference bus. 

In the presence of bad data and/or cyber attacks, observability analysis may not suffice. Consider the noiseless measurement model $\bz =\bh(\bv)+\bo$, where the non-zero entries of vector $\bo$ correspond to bad data or compromised meters; and let us proceed with the following definitions.

\begin{definition}[Observable attack~\cite{robust2012zhu}]\label{def:obs}
The attack vector $\bo$ is deemed as observable if for every state $\bv$ there is no $\bv'\neq \bv$, such that $\bh(\bv)+\bo=\bh(\bv')$. 
\end{definition}

\begin{definition}[Identifiable attack~\cite{robust2012zhu}] \label{def:id}
The attack vector $\bo$ is identifiable if for every $\bv$ there is no $(\bv',\bo')$ with $\bv'\neq \bv$ and $\|\bo'\|_0\leq \|\bo\|_0$, such that $\bh(\bv)+\bo=\bh(\bv')+\bo'$. 
\end{definition}

If the outlier vector $\bo$ is observable, the operator can tell that the collected measurements do not correspond to a system state, and can hence decide that an attack has been launched. Nevertheless, the attacked meters can be pinpointed only under the stronger conditions of Definition~\ref{def:id}. 

The resilience of the measurement mapping $\bh(\bv)$ against attacks can be characterized through $D(\bh)$ in \eqref{eq:mdf}: The maximum number of counterfeited meters for an attack to be observable is $K_o=D(\bh)-1$ and to be identifiable, it is $K_i=\lfloor \tfrac{D(\bh)-1}{2}\rfloor$; see \cite{robust2012zhu,tsp2013xumeng}. Here, the floor function $\lfloor x \rfloor$ returns the greatest integer less than or equal to $x$. 

Consider the linear mapping $\bh(\bv)=\bH\bv$. Measurement $m$ is termed \emph{critical} if once removed from the measurement set, it renders the power system non-identifiable. In other words, although $\bH$ is full column-rank, its submatrix $\bH_{(m)}$ is not. It trivially follows that $D(\bh)=1$, and the system operator can be arbitrarily misled even if only measurement $m$ is attacked. Due to the typically sparse structure of $\bH$, critical measurements or multiple simultaneously corrupted data do exist~\cite{AburExpositoBook}. It was pointed out in \cite{LiNi11} that if an attack $\bo$ can be constructed to lie in the $\range(\bH)$, it comprises a `stealth attack.' Although finding $D(\bh)$ is not trivial in general, a polynomial-time algorithm leveraging a graph-theoretic approach is devised in~\cite{Kosut11}.

\section{Power System State Tracking}\label{sec:dynamic}
The PSSE methods reviewed so far ignore system dynamics and do not exploit historical information. Dynamic PSSE is well motivated thanks to its improved robustness, observability, and predictive ability when additional temporal information is available~\cite{huang2002dynamic}. Recently proposed model-free and model-based state tracking schemes are outlined next.

\subsection{Model-free State Tracking via Online Learning}\label{subsec:dynamic:model-free}
In complex future power systems, one may not choose to explicitly commit to a model for the underlying system dynamics. The framework of online convex optimization (OCO), particularly popular in machine learning, can account for unmodeled dynamics and is thus briefly presented next~\cite{OCO}.

The OCO model considers a multi-stage game between a player and an adversary. In the PSSE context, the utility or the system operator assumes the role of the player, while the loads and renewable generations can be viewed as the adversary. At time $t$, the player first selects an action $\bV_t$ from a given action set $\mathcal{V}$, and the adversary subsequently reveals a convex loss function $f_t:\mathcal{V}\rightarrow\mathbb{R}$. In this round, the player suffers a loss $f_t(\mathbf{V}_t)$. The ultimate goal for the player is to minimize the \emph{regret} $R_f(T)$ over $T$ rounds:
\begin{equation}\label{eq:regret}
R_f(T):=\sum_{t=1}^T f_t(\mathbf{V}_t)-\underset{\mathbf{V}\in\mathcal{V}}{\minimize}\;\sum_{t=1}^Tf_t(\mathbf{V}).
\end{equation}    
The regret is basically the accumulated cost incurred by the player relative to that by a single fixed action $\mathbf{V}^0:=\arg\min_{\mathbf{V}\in\mathcal{V}}\sum_{t=1}^Tf_t(\mathbf{V})$. This fixed action is selected with the advantage of knowing the loss functions $\{f_t\}_{t=1}^T$ in hindsight. Under appropriate conditions, judiciously designed online optimization algorithms can achieve sublinear regret; that is, $R_f(T)/T\to 0$ as $T\to+\infty$.  

Building on the SDR-PSSE formulation of Section~\ref{subsec:problem:SDR}, the ensuing method considers streaming data for real-time PSSE. The data referring to and collected over the control period $t$ are $\{(z_{m_t};\mathbf{H}_{m_t})\}_{m_t=1}^{M_t}$ with $t=1,\ldots,T$. The number and type of measurements can change over time, while the matrix corresponding to measurement $m$ may change over time as indicated by $\{\mathbf{H}_{m_t}\}_{m_t=1}^{M_t}$ due to topology reconfigurations. The online PSSE task can be now formulated as
\begin{equation}\label{eq:opsse}
\underset{\mathbf{V}\succeq \mathbf{0}}{\minimize}~\sum_{t=1}^T f_t(\mathbf{V})
\end{equation}
where $f_t(\mathbf{V}):=\sum_{m_t=1}^{M_t}[z_{m_t}-\trace(\mathbf{H}_{m_t}\mathbf{V})]^2$. Online PSSE aims at improving the static estimates by capitalizing on previous measurements as well as tracking slow time-varying variations in generation and demand.

Minimizing the cost in \eqref{eq:opsse} may be computationally cumbersome for real-time implementation. An efficient alternative based on online gradient descent amounts to iteratively minimizing a regularized first-order approximation of the instantaneous cost instead~\cite{icassp2014skgwgg}
\begin{equation}\label{eq:ocoupdate}
\mathbf{V}_{t+1}:=\arg\min_{\mathbf{V}\succeq \mathbf{0}}~\trace(\mathbf{V}^\cH\nabla f_t(\mathbf{V}_t))+\frac{1}{2\mu_t}\left\|\mathbf{V}-\mathbf{V}_t\right\|_F^2
\end{equation}
for $t=1,\ldots$, and suitably selected step sizes $\mu_t>0$. 
Interestingly, 
the optimization in \eqref{eq:ocoupdate} admits a closed-form solution given by
\begin{equation}\label{eq:onpsse}
\mathbf{V}_{t+1}={\rm Proj}_{\mathbb{S}^+} [\mathbf{V}_t-\mu_t\nabla f_t(\mathbf{V}_t)]
\end{equation}
with ${\rm Proj}_{\mathbb{S}^+} $ denoting the projection onto the positive semidefinite cone, which can be performed using eigen-decomposition followed by setting negative eigenvalues to zero.    
It is worth mentioning that the online PSSE in \eqref{eq:onpsse}
enjoys sublinear regret~\cite{icassp2014skgwgg}. 
Upon finding $\mathbf{V}_t$, a state estimate $\mathbf{v}_t$ can be obtained by eigen-decomposition or randomization as in Section~\ref{subsec:problem:SDR}. With an additional nuclear-norm regularization term promoting low-rank solutions in \eqref{eq:ocoupdate}, online ADMM alternatives were devised in \cite{kim2015online}.
Interestingly, online learning tools has recently been advocated for numerous real-time energy management tasks in \cite{asilomar2014kim}, \cite{ tps2015vassilis}, \cite{tps2015wang}.

\subsection{Model-based State Tracking}\label{subsec:dynamic:model-based}
Although the previous model-free solver can recover slow time-varying states, model-based approaches facilitate tracking of fast time-varying system states. A typical state-space model for power system dynamics is~\cite{valverde2011unscented}
\begin{subequations}\label{eq:ssm}
	\begin{align}
	\mathbf{v}_{t+1}&=\mathbf{F}_t\mathbf{v}_t+\mathbf{g}_t+\bm{\omega}_t\label{eq:dynamic}\\
	\mathbf{z}_t&=\mathbf{h}(\mathbf{v}_t)+\bm{\epsilon}_t\label{eq:meas}
	\end{align}
\end{subequations}
where $\mathbf{F}_t$ denotes the state-transition matrix, $\mathbf{g}_t$ captures the process mismatch, 
and $\bm{w}_t$ is the additive noise. The nonlinear mapping $\mathbf{h}(\cdot)$ comes from conventional SCADA measurements. 
Values $\{(\mathbf{F}_t,\mathbf{g}_t)\}$ can be obtained in real-time using for example Holt's system identification method \cite{holt}. Two common dynamic tracking approaches to cope with the nonlinearity in the measurement model of \eqref{eq:meas} include the (extended or unscented) Kalman filters and moving horizon estimators \cite{valverde2011unscented,debs1970dynamic,huang2002dynamic,wang2012alternative,pesgm2014gwskgg}, and they are outlined in order next.

The extended Kalman filter (EKF) handles the nonlinearity by linearizing $\mathbf{h}(\mathbf{v})$ around the state predictor. To start, let $\hat{\mathbf{v}}_{t+1|t}$ stand for the predicted estimate at time $t+1$ given measurements $\{\mathbf{z}_{\tau}\}_{\tau=1}^t$ up to time $t$. Let also $\hat{\mathbf{v}}_{t+1|t+1}$ be the filtered estimate given measurements $\{\mathbf{z}_{\tau}\}_{\tau=1}^{t+1}$. If the noise terms $\bm{\omega}_t$ and $\bm{\epsilon}_t$ in \eqref{eq:ssm} are assumed zero-mean Gaussian with known covariance matrices $\mathbf{Q}_t\succeq\mathbf{0}$ and $\mathbf{R}_t\succeq\mathbf{0}$, respectively, the EKF can be implemented with the following recursions
\begin{equation}\label{eq:ekfrecursion}
\hat{\mathbf{v}}_{t+1|t+1}=\hat{\mathbf{v}}_{t+1|t}+\mathbf{K}_{t+1}\big[\mathbf{z}_{t+1}-\mathbf{h}(\hat{\mathbf{v}}_{t+1|t})
\big]
\end{equation}
where the state predictor $\hat{\mathbf{v}}_{t+1|t}$ and  
the Kalman gain $\mathbf{K}_{t+1}$ are given by
\begin{subequations}\label{eq:ekf}
	\begin{align}
	\hat{\mathbf{v}}_{t+1|t}&=\mathbf{F}_{t}\hat{\mathbf{v}}_{t|t}+\mathbf{g}_{t}\\
	\mathbf{K}_{t+1}&=\mathbf{P}_{t+1|t}\mathbf{J}_{t+1}^\mathcal{H}\big(\mathbf{J}_{t+1}\mathbf{P}_{t+1|t}\mathbf{J}_{t+1}^\mathcal{H}+\mathbf{R}_{t+1}\big)^{-1}\\
	\mathbf{P}_{t+1|t+1}&=\mathbf{P}_{t+1|t}-\mathbf{K}_{t+1}\mathbf{J}_{t+1}\mathbf{P}_{t+1|t}\\
	\mathbf{P}_{t+1|t}&=\mathbf{F}_{t}\mathbf{P}_{t|t}\mathbf{F}_{t}^\mathcal{H}+\mathbf{Q}_{t}
	\end{align}
\end{subequations}
with $\mathbf{J}_{t+1}$ being the measurement Jacobian matrix of $\mathbf{h}$
evaluated at $\hat{\mathbf{v}}_{t+1|t}$,
and $\mathbf{P}_{t+1|t+1}\succeq\mathbf{0}$ ($\mathbf{P}_{t+1|t}\succeq\mathbf{0}$)
denoting the corrected (predicted) state estimation error covariance matrix at time $t+1$. To improve on the approximation accuracy of the EKF, extended Kalman filters (UKF) have been reported in \cite{valverde2011unscented}; see also \cite{ZhaoNettoMili16} for their robust versions. Particle filtering may also be useful if its computational complexity can be supported during real-time power systems operations. 

Because the EKF and UKF are known to diverge for highly nonlinear dynamics, \emph{moving horizon estimation} (MHE) has been suggested as an accurate yet tractable alternative with proven robustness to bounded model errors~\cite{tac2003rao}. Different from Kalman filtering, the initial state $\mathbf{v}_0$, and noises $\bm{\omega}_t$ and $\bm{\epsilon}_t$ in MHE are viewed as deterministic unknowns taking values from given bounded sets $\mathcal{S}$, $\mathcal{W}$, and $\mathcal{E}$, respectively. The sets $\mathcal{W}$ and $\mathcal{E}$ model disturbances with truncated densities~\cite{tac2003rao}. 

The idea behind MHE is to perform PSSE by exploiting useful information present in a sliding window of the most recent observations. Consider here a sliding window of length $L+1$. Let $\hat{\mathbf{v}}_{t-L|t}$ denote the smoothed estimate at time $t-L$ given $L$ past measurements, as well as the current one, namely $\{\mathbf{z}_{\tau}\}_{\tau=t-L}^t$. MHE aims at obtaining the most recent $L$ state estimates $\{\hat{\mathbf{v}}_{t-L+s|t}\}_{s=0}^{L}$ based on $\{\mathbf{z}_{\tau}\}_{t-L}^t$ and the available estimate $\check{\mathbf{v}}_{t-L}:=\hat{\mathbf{v}}_{t-L|t-1}$ from time $t-1$ and for $t\ge L$. A key simplification is that once $\hat{\mathbf{v}}_{t-L|t}$ becomes available, the other $L$ recent estimates at time $t$ can be recursively obtained through `noise-free' propagation based on the dynamic model \eqref{eq:dynamic}; that is, 
\begin{equation}
\label{eq:dynamicupdate}
\hat{\mathbf{v}}_{t-L+s|t}=\mathbf{F}_{t-L+s-1}\hat{\mathbf{v}}_{t-L+s-1|t}
\end{equation}
for $s=1,\ldots,L$. By relating all recent estimates to $\hat{\mathbf{v}}_{t-L|t}$ via successive multiplications of transition matrices, the update in \eqref{eq:dynamicupdate} simplifies to
\begin{equation}
\hat{\mathbf{v}}_{t-L+s|t}=\mathbf{T}_{t-L+s}\hat{\mathbf{v}}_{t-L|t}
\end{equation}
where $\mathbf{T}_{t-L+s}:=\mathbf{F}_{t-L+s-1}\mathbf{T}_{t-L+s-1}$ for $s=1,\ldots,L$, with $\mathbf{T}_{t-L}=\mathbf{I}$. The MHE-based state estimate $\hat{\mathbf{v}}_{t-L|t}$ is then given by
\begin{equation}\label{eq:mhe}
\hat{\mathbf{v}}_{t-L|t}:=\arg\min_{\mathbf{v}}~\sum_{s=0}^L\big\|\mathbf{z}_{t-L+s}-\mathbf{h}(\mathbf{T}_{t-L+s}\mathbf{v})\big\|_2^2+\lambda\|\mathbf{v}-\check{\mathbf{v}}_{t-L}\|_2^2 
\end{equation}
where $\lambda>0$ can be tuned relying on our confidence in the state predictor $\check{\mathbf{v}}_{t-L}$, and the measurements $\{\mathbf{z}_{\tau}\}_{t-L}^t$. Given the quadratic dependence of the SCADA measurements $\{\mathbf{h}(\mathbf{v}_t)\}$ and the state $\mathbf{v}$, the optimization problem in \eqref{eq:mhe} is non-convex.

Finding the MHE-based state estimates in real time entails online solutions of dynamic optimization problems. The MHE formulation can be convexified by exploiting the semidefinite relaxation: vector $\mathbf{v}$ is lifted to the matrix $\mathbf{V}:=\mathbf{v}\mathbf{v}^\mathcal{H}\succeq\mathbf{0}$, and the $m$-th entry of $\mathbf{h}(\mathbf{v}_{t-L+s})$ for $s=0,\ldots,L$, is expressed as
\begin{equation*}
h_m(\mathbf{T}_{t-L+s}\mathbf{v})=\mathbf{v}^\mathcal{H}
\mathbf{T}_{t-L+s}^\mathcal{H}\mathbf{H}_m\mathbf{T}_{t-L+s}\mathbf{v}=  \trace(\mathbf{T}_{t-L+s}^\mathcal{H}\mathbf{H}_m\mathbf{T}_{t-L+s}\mathbf{V}).
\end{equation*}
Upon dropping the nonconvex rank constraint ${\rm rank}(\mathbf{V})=1$, the SDP-based MHE yields
\begin{equation*}
\hat{\mathbf{V}}_{t-L|t}:=\arg\min_{\mathbf{V}\succeq\mathbf{0}}~\sum_{s=0}^L\big\|\mathbf{z}_{t-L+s}-\trace\big(\mathbf{T}_{t-L+s}^\mathcal{H}\mathbf{H}_m\mathbf{T}_{t-L+s}\mathbf{V}\big)\big\|_2^2+\lambda\|\mathbf{v}-\check{\mathbf{v}}_{t-L}\|_2^2
\end{equation*}
which can be solved in polynomial time using off-the-shelf toolboxes. Rank-one state estimates can be obtained again through eigen-decomposition or randomization. The complexity of solving the last problem is rather high in its present form on the order of $N_b^{4.5}$~\cite{LuMa10}. Therefore, developing faster solvers for the SDP-based MHE by exploiting the rich sparsity structure in $\{\mathbf{H}_m\}$ matrices is worth investigating. Decentralized and localized MHE implementations are also timely and pertinent. Devising FPP-based solvers for the MHE in \eqref{eq:mhe} constitutes another research direction.

\section{Discussion}\label{sec:discussion}
This chapter has reviewed some of the recent advances in PSSE. After developing the CRLB, an SDP-based solver, and its regularized counterpart were discussed. To overcome the high complexity involved, a scheme named feasible point pursuit relying on successive convex approximations was also advocated. A decentralized PSSE paradigm put forth provides the means for coping with the computationally-intensive SDP formulations, it is tailored for the interconnected nature of modern grids, while it can also afford processing PMU data in a timely fashion. A better understanding of cyber attacks and disciplined ways for decentralized bad data processing were also provided. Finally, this chapter gave a fresh perspective to state tracking under model-free and model-based scenarios. 

Nonetheless, there are still many technically challenging and practically pertinent grid monitoring issues to be addressed. Solving power grid data processing tasks \emph{on the cloud} has been a major trend to alleviate data storage, communication, and interoperability costs for system operators and utilities. Moreover, with the current focus on low- and medium-voltage distribution grids, solvers for unbalanced and multi-phase operating conditions \cite{modeling2017gatsis} are desirable. Smart meters and synchrophasor data from distribution grids (also known as micro-PMUs \cite{micro2014}) call for new data processing solutions. Advances in machine learning and statistical signal processing, such as sparse and low-rank models, missing and incomplete data, tensor decompositions, deep learning, nonconvex and stochastic optimization tools, and (multi)kernel-based learning to name a few, are currently providing novel paths to grid monitoring tasks while realizing the vision of smarter energy systems.

\section*{Acknowledgements}
G. Wang and G. B. Giannakis were supported in part by NSF grants 1423316, 1442686, 1508993, and 1509040; and by
the Laboratory Directed Research and Development Program at the NREL. 
H. Zhu was supported in part by NSF grants 1610732 and 1653706.

\section*{Appendix}\label{sec:appendix}
\begin{proof}[Proof of Prop.~\ref{prop:crlb}]
	Consider the AGWN model \eqref{eq:WLS2} with $\bm{\epsilon}\sim \mathcal{N}(\mathbf{0},\diag(\{\sigma_m^2\}))$. The data likelihood function is 
	\begin{equation}\label{eq:likelihood}
	p(\mathbf{z};\mathbf{v})=\prod_{m=1}^M\frac{1}{\sqrt{2\pi \sigma_m^2}}
	\exp\!\left[-\frac{(z_m-\mathbf{v}^\mathcal{H}\mathbf{H}_m\mathbf{v})^2}{2\sigma_m^2}\right]
	\end{equation}
	and the negative log-likelihood function denoted by $f(\mathbf{v})=-\ln p(\mathbf{z};\mathbf{x})$ is
	\begin{equation}\label{eq:negative}
	f(\mathbf{v})=\sum_{m=1}^M\left[\frac{1}{2\sigma_m^2}\left(z_m-\mathbf{v}^\mathcal{H}\mathbf{H}_m\mathbf{v}\right)^2+\frac{1}{2}\ln\left(2\pi\sigma_m^2\right)\right].
	\end{equation}
	
	The Fisher information matrix (FIM) is defined as the Hessian of the real-valued function $f(\mathbf{v})$ with respect to the complex vector $\mathbf{v}\in\mathbb{C}^{N_b}$. Deriving the CRLB amounts to finding the Hessian of a real-valued function with respect to a complex-valued vector. \emph{Wirtinger's calculus} confirms that $f(\mathbf{v})$ can be equivalently rewritten as $f(\mathbf{v},{\mathbf{v}}^\ast)$; see e.g.,~\cite{wirtinger}. Upon introducing the conjugate coordinates $[\mathbf{v}^\mathcal{T}~({\mathbf{v}}^\ast)^\mathcal{T}]^\mathcal{T}\in\mathbb{C}^{2N_b}$, the \emph{Wirtinger derivatives}, namely the first-order partial differential operators of functions over complex domains, are given by \cite{wirtinger}
	\begin{subequations}
		\begin{align*}
		\frac{\partial f}{\partial \mathbf{v}}&:=\left.\frac{\partial f(\mathbf{v},{\mathbf{v}}^\ast)}{\partial\mathbf{v}^\mathcal{T}}\right|_{{{\rm constant}~\mathbf{v}}^\ast}
		=\left.\left[\frac{\partial f}{\partial \cV_1}~\cdots~\frac{\partial f}{\partial \cV_N}\right]\right|_{{{\rm constant}~\mathbf{v}}^\ast}\\
		\frac{\partial f}{\partial {\mathbf{v}}^\ast}&:=\left.\frac{\partial f(\mathbf{v},{\mathbf{v}}^\ast)}{\partial({\mathbf{v}}^\ast)^\mathcal{T}}\right|_{{\rm constant}~\mathbf{v}}=\left.\left[\frac{\partial f}{\partial \cV_1^\ast}~\cdots~\frac{\partial f}{\partial \cV_N^\ast}\right]\right|_{{{\rm constant}~\mathbf{v}}}.
		\end{align*}	
	\end{subequations}
These definitions follow the convention in multivariate calculus that derivatives are denoted by row vectors, and gradients by column vectors. Define for notational brevity $\phi_m (\mathbf{v}, {\mathbf{v}}^\ast): = z_m - ({\mathbf{v}}^\ast)^\mathcal{T}\mathbf{H}_m\mathbf{v}$ for $m=1,\ldots,M$. Accordingly, the Wirtinger derivatives of $f(\mathbf{v},\mathbf{v}^\ast)$ in \eqref{eq:negative} are obtained as
	\begin{equation}\label{eq:derivatives}
	\frac{\partial f}{\partial \mathbf{v}}  = \sum_{m = 1}^{M} \frac{1}{\sigma_m^2} \phi_m  \frac{\partial \phi_m }{\partial \mathbf{v}^\mathcal{T}}~~\text{and}~~
	\frac{\partial f}{\partial {\mathbf{v}^\ast}}=\sum_{m= 1}^{L} \frac{1}{\sigma_m^2} \phi_m \frac{\partial \phi_m }{\partial ({\mathbf{v}}^\ast)^\mathcal{T}}
	\end{equation}
	and the Wirtinger derivatives of $ \phi_m(\mathbf{v}, {\mathbf{v}}^\ast) $ can be found likewise
	\begin{align}
	\frac{\partial\phi_m}{\partial \mathbf{v}^\mathcal{T}} =-(\mathbf{H}_m \mathbf{{v}})^\mathcal{H}~~\text{and}~~
	\frac{\partial\phi_m}{\partial ({\mathbf{v}}^\ast)^\mathcal{T}} = -({\mathbf{H}}_m^\ast {\mathbf{v}}^\ast)^\mathcal{H}.\label{eq:inner}
	\end{align}
	
	In the conjugate coordinate system, the complex Hessian of $f(\bv,\bv^\ast)$ with respect to the conjugate coordinates $[\mathbf{v}^\mathcal{T}~({\mathbf{v}}^\ast)^\mathcal{T}]^\mathcal{T}$  
	is defined as
	\begin{equation}\label{eq:hessian}
	{\boldsymbol{\mathcal{H}}}(\bv,\bv^\ast):=\nabla^2 f(\mathbf{v}, {\mathbf{v}}^\ast)=
	\left[\begin{array}{cc}	{\boldsymbol{\mathcal{H}}}_{\mathbf{v}\mathbf{v}}&{\boldsymbol{\mathcal{H}}}_{{\mathbf{v}}^\ast\mathbf{v}}\\
	{\boldsymbol{\mathcal{H}}}_{\mathbf{v}{\mathbf{v}^\ast}}&{\boldsymbol{\mathcal{H}}}_{{\mathbf{v}}^\ast{\mathbf{v}}^\ast}
	\end{array}\right]
	\end{equation}
	whose blocks are given as
	\begin{align*}
		{\boldsymbol{\mathcal{H}}}_{\mathbf{v}\mathbf{v}}&:=\frac{\partial}{\partial \mathbf{v}^\mathcal{T}}\left(\frac{\partial f}{\partial \mathbf{v}}
	\right)^\mathcal{H}, \quad~~\, {\boldsymbol{\mathcal{H}}}_{{\mathbf{v}^\ast}\mathbf{v}}:=\frac{\partial}{\partial ({\mathbf{v}}^\ast)^\mathcal{T}}
	\left(\frac{\partial f}{\partial \mathbf{v}}
	\right)^\mathcal{H}\\
		{\boldsymbol{\mathcal{H}}}_{\mathbf{v}{\mathbf{v}^\ast}}&:=\frac{\partial}{\partial \mathbf{v}^\mathcal{T}}\left(\frac{\partial f}{\partial {\mathbf{v}^\ast}}
	\right)^\mathcal{H},\quad
	{\boldsymbol{\mathcal{H}}}_{{\mathbf{v}}^\ast{\mathbf{v}^\ast}}:=\frac{\partial}{\partial ({\mathbf{v}}^\ast)^\mathcal{T}}\left(\frac{\partial f}{\partial {\mathbf{v}}^\ast} \right)^\mathcal{H}.
	\end{align*}
	Substituting \eqref{eq:derivatives} and \eqref{eq:inner} into the last equations and after algebraic manipulations yields
	\begin{subequations}\label{eq:Hessianblks}
		\begin{align}
		{\boldsymbol{\mathcal{H}}}_{\mathbf{v}\mathbf{v}} 
		&= \sum_{m = 1}^{M} \sigma_m^{-2} \Big(\mathbf{H}_m \mathbf{v} (\mathbf{H}_m \mathbf{v})^\mathcal{H}- \phi_m \mathbf{H}_m\Big)\label{eq:blk11}\\
		{\boldsymbol{\mathcal{H}}}_{{\mathbf{v}}^\ast\mathbf{v}}
		&=\sum_{m = 1}^{M} \sigma_m^{-2}  \mathbf{H}_m \mathbf{v} ({\mathbf{H}}_m^\ast {\mathbf{v}}^\ast)^\mathcal{H}\label{eq:blk12}\\
		{\boldsymbol{\mathcal{H}}}_{\mathbf{v}{\mathbf{v}}^\ast}
		&= \sum_{m = 1}^{M} \sigma_m^{-2}  {\mathbf{H}}^\ast_m {\mathbf{v}}^\ast(\mathbf{H}_m \mathbf{v})^\mathcal{H}\label{eq:blk21}\\
		{\boldsymbol{\mathcal{H}}}_{{\mathbf{v}}^\ast{\mathbf{v}}^\ast}& =\sum_{m= 1}^{M} \sigma_m^{-2} \Big({\mathbf{H}}^\ast_m {\mathbf{v}}^\ast ({\mathbf{H}}^\ast_m {\mathbf{v}}^\ast)^\mathcal{H}- \phi_m  {\mathbf{H}}^\ast_m\Big)\label{eq:blk22}.
		\end{align}
	\end{subequations}
	Evaluating the Hessian blocks of \eqref{eq:Hessianblks} at the true value of $\mathbf{v}$, and taking the expectation with respect to $\bm{\epsilon}$, yields $ \mathbb{E} [ \phi_m ]=0$. Hence, the $\phi_m$-related terms in \eqref{eq:Hessianblks} disappear, and the FIM $\mathbf{F}(\bv,\bv^\ast) := \mathbb{E} [ {\boldsymbol{\mathcal{H}}}(\bv,\bv^\ast)]$ simplifies to the expression in \eqref{eq:fim}; see also~\cite{van1994cramer}.
	
	To show that the FIM is rank-deficient, define $\mathbf{g}_m:= [(\mathbf{H}_m\mathbf{v})^\mathcal{H}~~(\mathbf{H}_m^\ast\mathbf{v}^\ast)^\mathcal{H}]^\mathcal{H}$, so that the FIM becomes $\mathbf{F}= \sum_{m=1}^{M} \sigma_m^{-2}\mathbf{g}_m\mathbf{g}_m^\mathcal{H}$. Observe now that the non-zero vector $\mathbf{d}(\bv):=[\mathbf{v}^\mathcal{T}~-(\mathbf{v}^\ast)^\mathcal{T}]^\mathcal{T}$ is orthogonal to $\mathbf{g}_m$ for $m=1,\ldots,M$; that is,
	\[\mathbf{g}_m^\mathcal{H}\mathbf{d}= \mathbf{v}^\mathcal{H}\mathbf{H}_m\mathbf{v}- (\mathbf{v}^\mathcal{H}\mathbf{H}_m\mathbf{v})^\ast =0.\]
	Based on the latter, it is not hard to verify that $\mathbf{F}\mathbf{d}=\bzero$, which proves that the null space of $\mathbf{F}$ is non-empty.
\end{proof}

\bibliographystyle{IEEEtranS}
\bibliography{myabrv,used}
\end{document}